\RequirePackage{fix-cm}
\documentclass[smallextended]{svjour3}       
\smartqed  
\usepackage{microtype}
\usepackage[linesnumbered,ruled,vlined]{algorithm2e}
\usepackage{graphicx}
\usepackage{floatrow}
\usepackage{verbatim}
\usepackage{amsmath}
\usepackage{multirow}
\usepackage{float}
\usepackage{subfig}
\usepackage{mathrsfs}

\usepackage{authordate1-4}

\begin{document}

\title{Flow Network Models for Online Scheduling Real-time Tasks on Multiprocessors
}


\author{Hyeonjoong Cho         \and
        Arvind Easwaran 
}


\institute{Hyeonjoong Cho \at
              Dept. of Computer and Information Science, Korea University, South Korea \\
              Tel.: +82-44-8601374\\
              \email{raycho@korea.ac.kr}           
           \and
           Arvind Easwaran \at
              School of Computer Engineering, Nanyang Technological University, Singapore
}

\date{Received: date / Accepted: date}

\maketitle

\begin{abstract}
We consider the \textit{flow network model} to solve the multiprocessor real-time task scheduling problems. Using the flow network model or its generic form, \textit{linear programming} (LP) formulation, for the problems is not new. However, the previous works have limitations, for example, that they are classified as \textit{offline} scheduling techniques since they establish a flow network model or an LP problem considering a very long time interval. In this study, we propose how to construct the flow network model for \textit{online} scheduling periodic real-time tasks on multiprocessors. Our key idea is to construct the flow network only for the active instances of tasks at the current scheduling time, while guaranteeing the existence of an optimal schedule for the future instances of the tasks. The optimal scheduling is here defined to ensure that all real-time tasks meet their deadlines when the total utilization demand of the given tasks does not exceed the total processing capacity. We then propose the flow network model-based polynomial-time scheduling algorithms. Advantageously, the flow network model allows the task workload to be collected unfairly within a certain time interval without losing the optimality. It thus leads us to designing three \textit{unfair-but-optimal} scheduling algorithms on both \textit{continuous} and \textit{discrete}-time models. Especially, our unfair-but-optimal scheduling algorithm on a discrete-time model is, to the best of our knowledge, the first in the problem domain. We experimentally demonstrate that it significantly alleviates the scheduling overheads, i.e., the reduced number of preemptions with the comparable number of task migrations across processors.
\end{abstract}

\section{Introduction}\label{sec:introduction}

Multicore or multiprocessor platforms are becoming prevalent in numerous digital devices and this advance has been accelerated by the increasing computational demands of various emerging high-quality services. Alongside this trend, there has been a vast amount of research into multiprocessor real-time scheduling theories~\cite{ra2011}\cite{bbb2015}. Real-time systems are computing systems where their correct behaviors depend not only on the value of the computation but also when the results are produced~\cite{b2011}. Most research problems related to multiprocessors involve far more than a simple theoretical extension from uniprocessors to multiprocessors and thus, the real-time scheduling problems on multiprocessors are challenging.

Liu stated~\cite{l1969}: \textit{``Few of the results obtained for a single processor generalize directly to the multiple processor case; bringing in additional processors adds a new dimension to the scheduling problem. The simple fact that a task can use only one processor even when several processors are free at the same time adds a surprising amount of difficulty to the scheduling of multiple processors''}. This statement can be interpreted as saying that in uniprocessors, the constraint that each task is forbidden from executing simultaneously on more than one processor is implicit, because a single processor is the only processing capacity present in the system. By contrast, in multiprocessors, the constraint of \textit{no intra-task parallelism} becomes not only explicit but also interrelated with the other constraints, which significantly increases the problem's complexity.

\subsection{Motivational examples}
\label{sec:motiv_examples}

We provide several examples in this section to specifically illustrate our research motivation. First, we introduce some of the basic notions and terminology used throughout this study.  

We assume that each task $\tau_i$ is characterized by $(C_i, D_i, P_i)$ where $C_i$, $D_i$ and $P_i$ are its worst-case execution time, deadline, and period, respectively. When $D_i$ is identical to $P_i$, it is called the \textit{implicit} deadline and $\tau_i$ is then characterized by $(C_i, P_i)$. The \textit{hyper-period} $H$ of all tasks is defined as the least-common-multiple of all $P_i$. The $j$th instance (or \textit{job}) of the periodic task $\tau_i$ is denoted by $\tau_{i,j}$ and the arrival time of $\tau_{i,j}$ is denoted by $a_{i,j}$. The \textit{running rate} is the ratio of the execution time of a job (or a part of the job) to the time interval for the execution. For example, when $\tau_{i,j}$ has its execution time $C_{i,j}$ within its deadline $D_i$, its the running rate $r_{i,j}$ is $C_{i,j}/D_i$. 

A set $\mathbf{B}$ containing all $a_{i,j}$ of the given jobs is called a set of \textit{boundaries}. For convenience, each element (or boundary) of $\mathbf{B}$ is denoted by $b_k$ where an earlier $b_k$ has a lower $k$. The individual utilization $u_i$ of $\tau_i$ is defined as $C_i/P_i$ and the total utilization $U$ of the given task set is the sum of all $u_i$. We assume that $N$ tasks run on $M$ processors.

\begin{definition} (\textit{RT-optimality}) An \textit{optimal} real-time schedule meets all the task deadlines when the total utilization demand $U$ of a given task set does not exceed the total processing capacity $M$, which we call \textit{RT-optimal} in this study. 
\end{definition}


Several classes of RT-optimal task scheduling algorithms on multiprocessors have been developed for periodic implicit-deadline tasks. One of the well-known classes is \textit{fluid schedule}-based algorithms, where each task execution attempts to track the \textit{fluid schedule} that is known to be RT-optimal~\cite{bcpv1993,shab2003}. Here, the fluid schedule is defined as follows.

\begin{definition} (Fluid schedule) A schedule is said to be fluid if and only if at any time $t \geq 0$, every instance $\tau_{i,j}$ of task $\tau_i$ that arrives at time $a_{i,j}$ has been executed for exactly $(t - a_{i,j}) \times C_i/P_i$ time units~\cite{ghydvj2014}.
\end{definition}

Thus, to achieve RT-optimality, the fluid schedule assigns each task $\tau_i$ with a uniform running rate $r_i$ for all times, where $r_i$ is set to $C_i/P_i$ (or $u_i$). The fluid schedule is known to be RT-optimal but unrealistic, since it allocates a fraction of the computational resource to each task during each time unit, e.g., 1/3 of the processing capacity is allocated to a task per unit time.

The first example shown in Table~\ref{tbl:example_1} includes five tasks running on two processors. Its fluid schedule is illustrated in Figure~\ref{fig:motiv_example_1}. Their total utilization demand is $U=2$, which is equal to the total processing capacity 2 and thus, RT-optimality of this fluid schedule holds. 




\begin{table}[h]
	\caption{Example task set for fluid and boundary-fair schedules}
	\label{tbl:example_1}
	\renewcommand{\arraystretch}{0.8}
	\begin{tabular}{|c|c|c|c|}
	\hline
	Task & $C_i$ & $P_i$ & $u_i$ \\
	\hline 
	$\tau_1$ & 2 & 3 & 2/3\\
	$\tau_2$ & 2 & 6 & 1/3\\
	$\tau_3$ & 2 & 6 & 1/3\\
	$\tau_4$ & 3 & 9 & 1/3\\
	$\tau_5$ & 3 & 9 & 1/3\\
	\hline
	\end{tabular}
\end{table}

\begin{figure}[h]
	\includegraphics[scale=0.65]{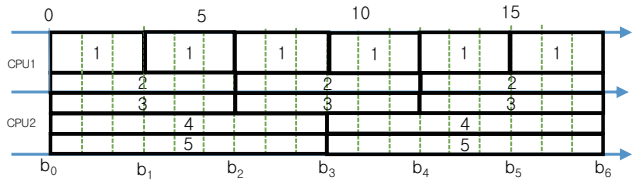}
 \caption{Fluid schedule for the tasks in Table~\ref{tbl:example_1}}
 \label{fig:motiv_example_1}
\end{figure}

One advantage of this fluid schedule is that it satisfies all of the interrelated constraints of the multiprocessor real-time scheduling problem, including the constraint of no intra-task parallelism, by using a single parameter $r_i$ per task. Provided that each $r_i$ is not greater than one and the sum of all $r_i$ is not greater than the total processing capacity, all tasks satisfy the deadlines without violating their constraints of no intra-task parallelism within the permitted total processing capacity. 

\begin{figure}[b]
\centering
\includegraphics[scale=0.65]{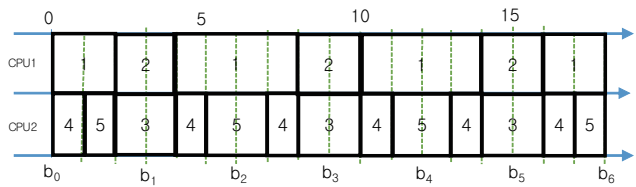}
\caption{Boundary-fair schedule for the tasks in Table~\ref{tbl:example_1}}
\label{fig:motiv_example_2}
\end{figure}

Several inheritors of the fluid schedule, which we refer to as the \textit{fluid schedule-based} algorithms, attempt to track the fluid schedule in order to obtain RT-optimality. To track the fluid schedule, they restrict the difference between the actual computational resource allocation and the fluid-schedule-based resource allocation for each task. The first algorithm of this type, proportionate-fair(\textit{Pfair}), strictly maintains the restriction at every time quantum~\cite{bcpv1993,bcpv1996}; its several descendants relax the restriction to being maintained at every boundary~\cite{zmm2003,crj2006,flspb2011,fky2008,zqmm2011}. Pfair and some descendants rely on a \textit{discrete-time model}, since they allocate the integral units of the computational-resource to each task, e.g., integer multiples of the system time unit for execution. In addition, they are known to support \textit{fairness} in the sense that the computational-resource allocated to each task is always proportionate to its $u_i$ during the time interval from zero to any time quantum for Pfair or to any boundary for its descendants.

Figure~\ref{fig:motiv_example_2} shows the task set in Table~\ref{tbl:example_1} scheduled by one of Pfair's descendants, \textit{boundary-fair scheduling (BF)}. It is observed that each task has the same amount of the computational resource within every time interval between two adjacent boundaries, i.e., the fairness is supported. For example, $\tau_1$ executes for 2 time units within every time interval [$b_k$,$b_{k+1}$].


Recently, it was observed that the scheduling overheads, including the number of preemptions and task migrations across processors, decrease as the fairness is relaxed~\cite{gvjd2011,nbngm2012}. Based on this observation, an unfair-but-optimal \textit{U-EDF} was proposed. The unfairness of U-EDF implies that some of the task workload is allowed to be advanced or delayed across boundaries beyond fairness. It implies that unlike the fluid schedule-based approaches, all tasks do not even need to run for every time interval between two adjacent boundaries and thus, the number of preemptions can be reduced. U-EDF relies on a \textit{continuous-time model}, since it may allocate a fractional computational-resource to a task, e.g., 1/3 processing capacity allocated to a task for a certain time interval.  


In addition to the reduced scheduling overheads, we believe that the \textit{unfairness} has several other significant advantages. For example, unfairness allows the total workload to be collected within certain time intervals and processors within the other idle intervals can turn into a different state to minimize their energy consumption. Figure~\ref{fig:motiv_example_3} shows an unfair schedule for five tasks in Table~\ref{tbl:example_2}. In this schedule, the workloads are aggregated around $b_3$ and thus, two processors have the chance to become slow or idle in both $[b_0,b_1]$ and $[b_5,b_6]$.


\begin{table}[h]
\renewcommand{\arraystretch}{0.8}
\caption{Example task set for slowing or idling processors}
\label{tbl:example_2}
\centering
\begin{tabular}{|c|c|c|c|}
\hline
Task & $C_i$ & $P_i$ & $u_i$ \\
\hline 
$\tau_1$ & 1 & 3 & 1/3\\
$\tau_2$ & 2 & 6 & 1/3\\
$\tau_3$ & 2 & 6 & 1/3\\
$\tau_4$ & 1 & 9 & 1/9\\
$\tau_5$ & 1 & 9 & 1/9\\
\hline
\end{tabular}
\end{table}

\begin{figure}[h]
\centering
\includegraphics[scale=0.65]{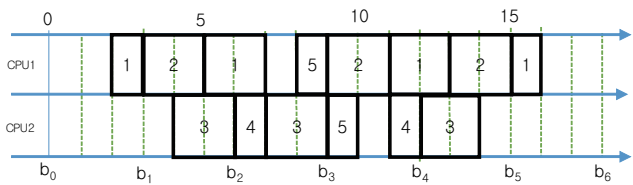}
\caption{Idling processors for the tasks in Table~\ref{tbl:example_2}}
\label{fig:motiv_example_3}
\end{figure}






\subsection{Contribution}

In order to support the unfair-but-optimal scheduling, we build a framework that allows us to manipulate the task workload efficiently across boundaries while holding RT-optimality, which is the primary focus of this study. Our contributions include the following: 

(1) We formulate the problem for online-scheduling the periodic implicit-deadline tasks on multiprocessors by specifying its constraints and we propose a \textit{flow network model} to solve the formulated problem. Again, using the flow network model or the LP formulation for multiprocessor real-time task scheduling is not new. However, the previous works have limitations, such as that some are used as offline scheduling techniques since they require a single flow network model or a single LP problem to be constructed considering a very long time interval from $0$ to $H$ of the given tasks~\cite{ec1981,msd2010,ljp2013} or that others are applicable only to the aperiodic tasks~\cite{h1974}. To overcome the limitations, we propose how to construct the flow network model only for the active instances of tasks at the current scheduling time while guaranteeing the existence of an RT-optimal schedule for the future instances of the tasks. 

(2) Based on the framework, we introduce an unfair-but-optimal multiprocessor scheduling algorithm, called \textit{flow network-based Earliest-Deadline-First} (fn-EDF), for periodic tasks on both continuous and discrete-time models. In particular, to the best of our knowledge, fn-EDF on a discrete-time model is the first online unfair-but-optimal scheduling algorithm in the given problem domain. We experimentally show that fn-EDF on the discrete-time model significantly reduces the number of preemptions with the comparable number of migrations against an existing BF algorithm. 

Table~\ref{tbl:comparison} compares fn-EDF with existing algorithms in terms of their problem domains. BF is fair-and-optimal on a discrete-time model and both U-EDF and RUN~\cite{rlmlb2011} are unfair-but-optimal on a continuous-time model. The original LP-based scheduling focused on constrained-deadline tasks, which have their deadlines less than their periods~\cite{ec1981,msd2010}. In Table~\ref{tbl:comparison}, it is assumed that LP-based scheduling can easily solve the implicit-deadline task scheduling problem by setting $D_i=P_i$. As discussed, the traditional LP-based scheduling is classified as an offline approach.
  
The remainder of this paper is organized as follows. Section~\ref{sec:fnm_implicit} presents the problem formulation and the corresponding flow network model for scheduling implicit-deadline periodic tasks; the unfair-but-optimal scheduling algorithms on both continuous and discrete-time models are also discussed. Section~\ref{sec:experiments} experimentally evaluate the performance of the proposed algorithm compared with an existing method. In Section~\ref{sec:discussion}, we discuss two issues about the proposed approach, complexity and extensibility. Section~\ref{sec:related} summarizes the related work. In Section~\ref{sec:conclusion} we give our conclusion.

\begin{table}[]
\renewcommand{\arraystretch}{0.9}
\caption{Comparison}
\label{tbl:comparison}
\centering
\begin{tabular}{|c|c|c|c|c|}
\hline
                                 &  \bf Algorithm     & \bf Continuous time model  & \bf Discrete time model &  \bf Reference  \\ \hline 
\multirow{5}{*}{\textbf{online}} & \multirow{2}{*}{fn-EDF} & RT-optimal  &  -         &  Section~\ref{subsec:flow_control} \\ \cline{3-5}
                                 &                         &  -          & RT-optimal &  Section~\ref{subsec:discreteness} \\ \cline{2-5}
                                 & BF                      &  -          & RT-optimal & \cite{zqmm2011} \\ \cline{2-5}
                                 & U-EDF                   & RT-optimal  & -          & \cite{nbngm2012} \\ \cline{2-5}
                                 & RUN                     & RT-optimal  & -          & \cite{rlmlb2011} \\ \hline
\bf offline                      & LP-based & RT-optimal  & RT-optimal &   \cite{ec1981,msd2010} \\
\hline
\end{tabular}
\end{table}

\section{Scheduling Algorithms}
\label{sec:fnm_implicit}

\subsection{System model}

We consider the implicit-deadline task $\tau_i$ characterized by $(C_i, P_i)$. We assume that the task parameters, $C_i$ and $P_i$, are multiples of the system time unit. The \textit{active} job $\tau_{i}(t)$ of $\tau_i$ at time $t$ has its arrival time $a_{i}(t)$ subject to $t \in [a_{i}(t),a_{i}(t)+P_i)$. At time $t$, $\tau_{i}(t)$ has its remaining execution time $c_i(t)$, where $0 \leq c_i(t) \leq C_i$. The active job set is denoted by $\boldsymbol{\tau}(t)$. Both $a_{i,j}+P_i$ and $a_{i}(t)+P_i$ correspond to $d_{i,j}$ and $d_i(t)$, respectively and called \textit{absolute deadlines} of the job. 

$b_{k}$ is the $k$th boundary of $\mathbf{B}$. The earlier $b_k$ is assumed to have the lower index $k$. When we consider the active jobs only, $\mathbf{B}=\{d_i(t)| 1 \leq i \leq N \} \cup \{t\}$ contains $K$ absolute deadlines (boundaries) of the active jobs, where $K$ is less than or equal to $N$. The window $W_k$ is the time interval ranging [$b_k$,$b_{k+1}$). The length $l_k$ of $W_k$ is $b_{k+1}-b_k$ and the permitted processing capacity $Cap(W_k)$ of $W_k$ is upper-bounded by $M \times l_k$. The allocated execution time for $\tau_{i}$ within $W_k$ is denoted by $X_{i,k}$.

\subsection{Problem formulation with an example}

Assume that 5 tasks in Table~\ref{tbl:example_1} are running on 2 processors. To consider active jobs at the current time $t$, we focus on the time interval from the current time $t$ to the latest deadline of all active jobs $d^{max}(t)=max_{\forall i}{\{d_{i}(t)\}}$. In $[t,d^{max}(t)]$, the boundary set is $\mathbf{B}=\{d_i(t)|1\leq i\leq N\} \cup \{t\}$. For the given example, all active jobs at current time 0 are illustrated in white in Figure~\ref{fig:step1} and $d^{max}(0)$ is 9. For the active jobs, $X_{i,k}$ are assigned as shown in Figure~\ref{fig:step1}.

\begin{figure}[h!]
\centerline{
	\includegraphics[scale=0.65]{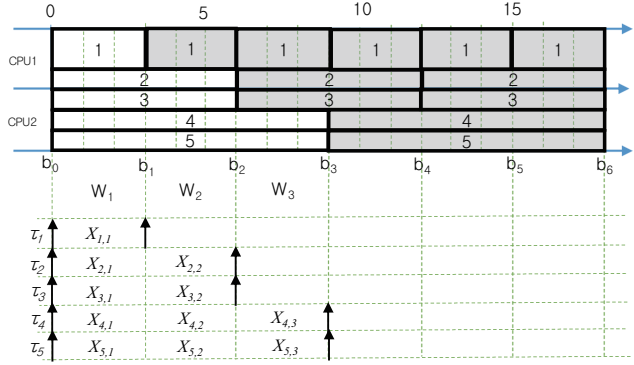}
	}
\caption{Active jobs at time 0}
\label{fig:step1}
\end{figure}

Next, we formulate a linear programming problem for the active jobs at time $0$.


Equation~\ref{eqt:jcc1} defines the constraint that each active job $\tau_{i}(t)$ completes its execution within the permitted time interval, i.e., $[a_{i}(t),d_{i}(t))$. We call this the \textit{job completion constraint (JCC)}.

\begin{equation} \label{eqt:jcc1}
\begin{split}
&X_{1,1} = C_1 \\
&X_{2,1}+X_{2,2} = C_2 \\
&X_{3,1}+X_{3,2} = C_3 \\
&X_{4,1}+X_{4,2} + X_{4,3} = C_4 \\ 
&X_{5,1}+X_{5,2} + X_{5,3} = C_5 \\ 
\end{split}
\end{equation}

Inequality~\ref{eqt:pcc1} is the constraint that the sum of the allocated execution times of the active jobs within a given window does not exceed the permitted processing capacity of the window. We call this the \textit{processing capacity constraint (PCC)}.

\begin{equation} \label{eqt:pcc1}
\begin{split}
X_{1,1}+X_{2,1}+X_{3,1}+X_{4,1}+X_{5,1} & \leq Cap(W_1) \\
X_{2,2}+X_{3,2}+X_{4,2}+X_{5,2} & \leq Cap(W_2)  \\
X_{4,3}+X_{5,3} & \leq Cap(W_3)  \\
\end{split}
\end{equation}

Inequality~\ref{eqt:ipc1} is the constraint that each active job within a given window does not simultaneously occupy more than one processor. We call this the constraint of \textit{no intra-task parallelism (NIP)}. 
	
\begin{equation} \label{eqt:ipc1}
\begin{split}
&X_{1,1}, X_{2,1}, X_{3,1}, X_{4,1}, X_{5,1} \leq l_1 \\
&X_{2,2}, X_{3,2}, X_{4,2}, X_{5,2} \leq l_2  \\
&X_{4,3}, X_{5,3} \leq l_3  \\
\end{split}
\end{equation}

To facilitate the discussion, the white region ranging $[t,d^{max}(t))$ in Figure~\ref{fig:step1} is called the \textit{active job area} at time $t$ and it is denoted by $AJA(t)$. In other words, $AJA(t)$ is a collection of the maximum capacity per window that can be utilized to execute the active jobs at $t$. For each $AJA(t)$, three types of constraints, JCC, PCC and NIP, are defined. No objective function is required and a feasible solution that satisfies all the constraints is sufficient for our purpose.


In the constraints above, all right-hand side values except $Cap(W_k)$ can be determined easily when $\mathbf{B}$ for the active jobs is fixed. The rationale behind the determination of $Cap(W_k)$ is explained using the following example. Figure~\ref{fig:step1} indicates that some capacity of $W_2$ is used for the active jobs, i.e., \{$\tau_{2,1}$, $\tau_{3,1}$, $\tau_{4,1}$, $\tau_{5,1}$\}, and the remainder is reserved for a future job $\tau_{1,2}$. Thus, each $Cap(W_k)$ should be determined in order to support the schedulabilities of the current active jobs and future jobs as well. Therefore, we determine $Cap(W_k)$ using the fluid schedule. For example, $Cap(W_3)$ is calculated by reserving the execution time of \{$\tau_{1,3}$, $\tau_{2,2}$, $\tau_{3,2}$\} within $W_3$ from the permitted processing capacity of $W_3$, e.g., $Cap(W_3)=(M-u_1-u_2-u_3)\times l_3$. When a schedulable implicit-deadline task set is given, since we reserve a suitable amount of processing capacity for future jobs based on the fluid schedule, the schedulability of future jobs remains valid regardless of our scheduling decision within the current AJA.

To schedule real-time tasks up to their hyper-period, we assume that AJA is constructed repeatedly at each boundary (scheduling event). Whenever AJA is constructed at time $t$, the right-hand sides of JCC equations are filled with the remaining execution time $c_i(t)$ of each $\tau_i(t)$. PCCs and NIPs are formed in the same manner as shown for time 0. Since a new AJA is constructed at each boundary, some $X_{i,k}$ in the previous AJA are recalculated. In this example, at time 0, only $X_{i,1}$ in $W_1$ are used for scheduling. The other allocated execution times $X_{i,k}$ are recalculated in the next AJA. 

Figure~\ref{fig:step2} shows both AJA(3) and AJA(6). We assume that the indices of $b_k$ and $W_k$ are updated for each AJA($t$). Once all $X_{i,k}$ in $W_k$ are obtained, the actual schedule within $W_k$ can easily be determined, e.g., using McNaughton's wrap around algorithm~\cite{m1959}.

\begin{figure}[h!]
\centering
	\subfloat{
		\includegraphics[scale=0.65]{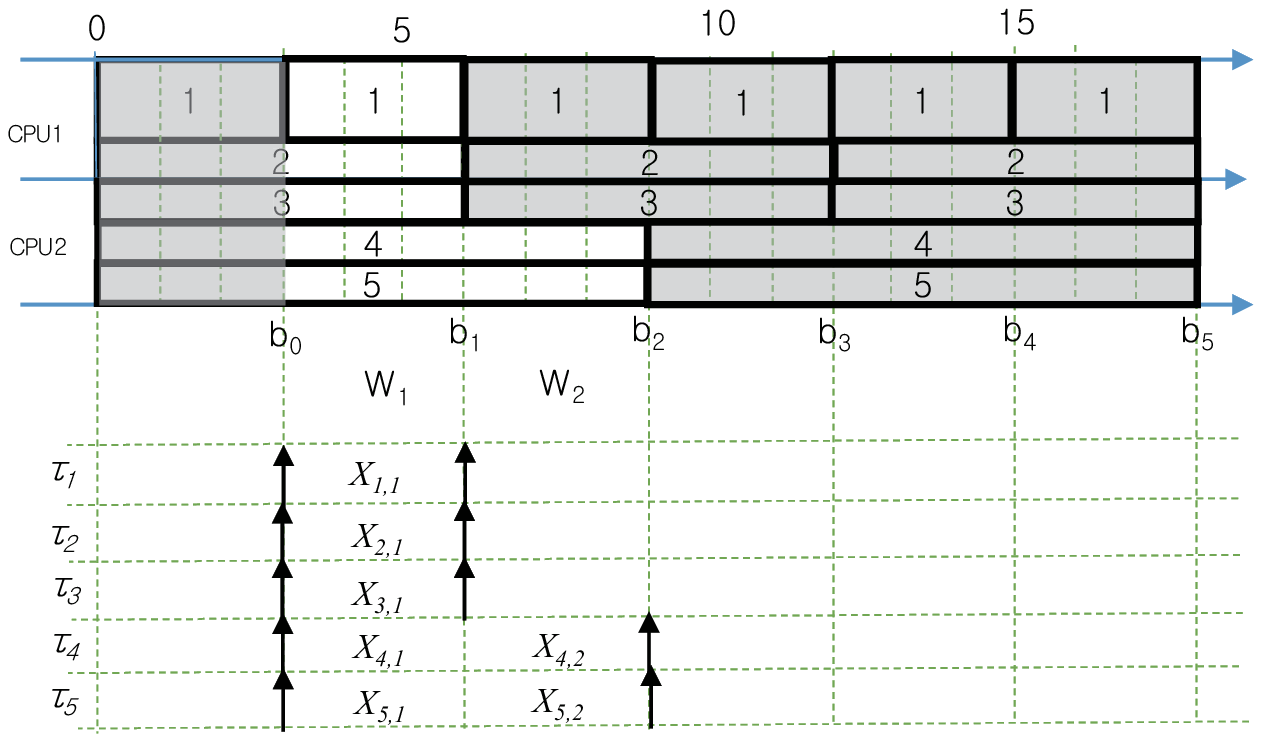}
	}\\
	\subfloat{
		\includegraphics[scale=0.65]{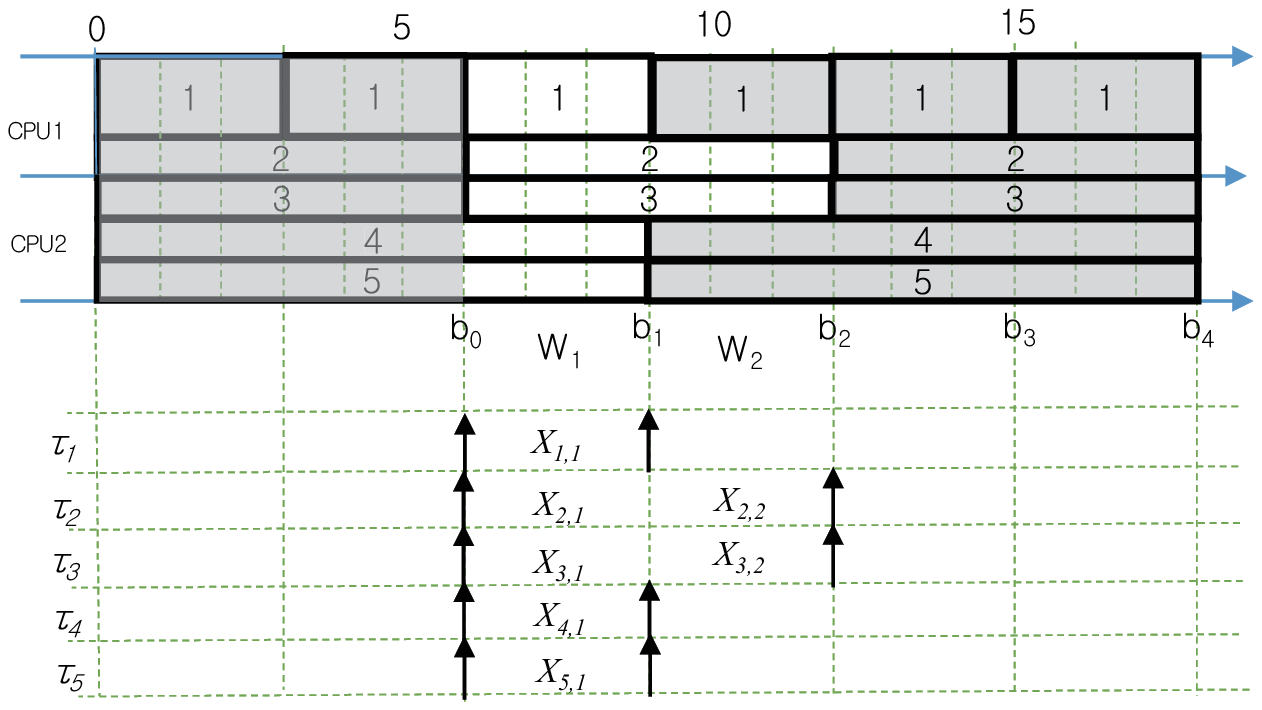}
	}	
\caption{Active jobs at time 3 and 6}
\label{fig:step2}
\end{figure}

We emphasize that a solution satisfying three types of constraints may yield the different running rates of a job between two adjacent windows, e.g., $r_{i,k} \neq r_{i,k+1}$, which implies that this approach can generate an unfair schedule.

\subsection{Problem formulation with the general task model}
\label{sec:formulation_implicit}

Figure~\ref{fig:general_implicit} shows a set of general tasks with implicit deadlines. At the current time $t_1$, $AJA(t_1)$ is constructed with three types of constraints. We set the current time $t_1$ to the boundary $b_0$. The first window $W_1$ then ranges from $b_0$ to $b_1$ (= $t_2$). The active job $\tau_{i}(t_1)$ is assumed to have its remaining execution time $c_{i}(t_1)$ at $t_1$. Within the time interval [$t_1$, $d^{max}$($t_1$)), a series of windows, \{$W_1$, ..., $W_K$\}, is built, based on $\mathbf{B}$. The number of windows $K$ for $AJA(t_1)$ is less than or equal to $N$.

\begin{figure}[h!]
\centerline{
	\includegraphics[scale=0.6]{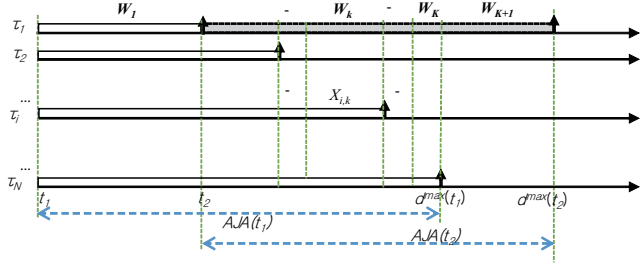}
	}
\caption{General task model with implicit deadlines and AJAs}
\label{fig:general_implicit}
\end{figure}

For convenience, we define two sets for the current AJA as follows:

\begin{equation} 
\boldsymbol{K}(t_s,t_e) = \{ k | W_k \subset [t_s, t_e] \},
\end{equation}

\begin{equation} 
\boldsymbol{J}(k, t) = \{ i | W_k \subset  [a_{i}(t), d_{i}(t)] \}\text{.}
\end{equation}

$\boldsymbol{K}(t_s,t_e)$ contains all indices of the window $W_k$ which are placed in the time interval [$t_s$,$t_e$]. $\boldsymbol{J}(k, t)$ contains all indices $i$ of the active jobs at time $t$ which are still active in $W_k$. Using these two sets, the constraints for $AJA(t_1)$ are defined as follows:

\begin{equation} \label{eqt:jcc_g}
\sum_{\forall k \in \boldsymbol{K}(t_1,d_{i}(t_1))}{X_{i,k}}=c_{i}(t_1), \quad 1 \leq \forall i \leq N,
\end{equation}

\begin{equation} \label{eqt:pcc_g}
\sum_{\forall i \in \boldsymbol{J}(k,t_1)}{X_{i,k}} \leq Cap(W_k), \quad 1 \leq \forall k \leq K,
\end{equation}

\begin{equation}  \label{eqt:ipc_g}
X_{i,k} \leq l_k, \quad 1 \leq \forall i \leq N, 1 \leq \forall k \leq K \text{.}
\end{equation}

For PCC, $Cap(W_k)$ is set as follows:

\begin{equation} \label{eqt:cap_g}
Cap(W_k)=\left[ M-\sum_{ \forall i \notin \boldsymbol{J}(k,t_1)}{C_i/P_i}\right] \times l_k \text{.}
\end{equation}

At a scheduling event (or boundary), the corresponding AJA is established and three types of constraints are defined. After a feasible solution is found for AJA, the part of the feasible solution for $W_1$ is used to allocate the computational resource to each task for their execution during $W_1$. At the next scheduling event, the next AJA is established and a similar procedure follows. This iteration continues until $H$, which provides the schedulability for the given tasks. Note that the LP formulation includes no objective function and a feasible solution that satisfies all constraints is sufficient for our purpose. 

RT-optimality of the proposed approach can be proved by showing that as long as the fluid schedule can be defined for a given task set, the proposed approach finds at least a feasible schedule.

\begin{lemma} \label{lem:1}
If a fluid schedule exists for the given periodic implicit-deadline task set, the fluid schedule satisfies the three types of constraints in the first AJA($0$).
\end{lemma}
\begin{proof}
The fluid schedule guarantees schedulability by providing each task with the processing capacity based on its running rate $r_i$ that is equivalent to its individual utilization $u_i=C_i/P_i \leq 1$. For AJA($0$), the fluid schedule ensures that every $X_{i,k}$ is $r_i \times l_k=u_i \times l_k$. First, JCCs for AJA($0$) are satisfied as follows. 
\begin{equation}
\sum_{\forall k \in \boldsymbol{K}(0,d_{i}(0))}{X_{i,k}} = u_i\times\sum_{\forall k \in \boldsymbol{K}(0,d_{i}(0))}{l_k}=u_iP_i=C_i.
\end{equation}
Second, NIPs for AJA($0$) are satisfied since $X_{i,k}=r_i \times l_k =u_i \times l_k \leq l_k$. 
Third, PCCs for AJA($0$) are satisfied as follows.
Since the fluid schedule is assumed, the following equation holds.
\begin{equation}
\begin{split}
\sum_{\forall i \in \boldsymbol{J}(k,0)}{X_{i,k}} + \sum_{ \forall i \notin \boldsymbol{J}(k,0)}{X_{i,k}} = l_k \times \sum_{i=1}^N{r_i} \\
\Leftrightarrow \sum_{\forall i \in \boldsymbol{J}(k,0)}{X_{i,k}} =l_k \times \sum_{i=1}^N{r_i} - \sum_{ \forall i \notin \boldsymbol{J}(k,0)}{X_{i,k}}
\end{split}
\end{equation}
By the assumption of the fluid schedule, $\sum_{i=1}^N{r_i} \leq M$ and $\sum_{ \forall i \notin \boldsymbol{J}(k,0)}{X_{i,k}}=\sum_{ \forall i \notin \boldsymbol{J}(k,0)}{u_{i}l_k}$. Thus, 
\begin{equation}
\sum_{\forall i \in \boldsymbol{J}(k,0)}{X_{i,k}} \leq l_k \times M - \sum_{ \forall i \notin \boldsymbol{J}(k,0)}{u_il_k}.
\end{equation}
From Equation~\ref{eqt:cap_g}, $\sum_{\forall i \in \boldsymbol{J}(k,0)}{X_{i,k}} \leq Cap(W_k)$. 
\end{proof}

\begin{theorem} \label{thm:optimality}
If a fluid schedule exists for the given periodic implicit-deadline task set, the proposed approach provides a feasible solution satisfying all constraints established by the task set.
\end{theorem}

\begin{proof}
The proof is obtained by induction on the increasing boundary. We assume that AJA(0) has a feasible solution that satisfies all constraints. Then, we try to show that if AJA($t_1$) has a feasible solution, AJA($t_2$) has at least one feasible solution, where $t_2$ is the next boundary to $t_1$ as shown in Figure~\ref{fig:general_implicit}. Note that the index $k$ is not updated at the next scheduling event point $t_2$ simply for convenience in the proof.\\
\textit{Basis.} From Lemma~\ref{lem:1}, AJA($t_1$) in Figure~\ref{fig:general_implicit} has at least one feasible solution, when $t_1=0$. Here, assume that the active jobs at $t_1$ are sorted in the increasing deadline order.\\
\textit{Induction step.}
Assume that $\boldsymbol{X^*}$ is a feasible solution for AJA($t_1$). After $\{X_{i,1}^{*} | 1 \leq \forall i \leq N \}$ is allocated for $W_1$, $\boldsymbol{X^{*}_{-1}}=\boldsymbol{X^*} \setminus \{X_{i,1}^{*} | 1 \leq \forall i \leq N \}$ is still a feasible solution for the remaining area of AJA($t_1$). Therefore, $\boldsymbol{X^{*}_{-1}}$ satisfies the following equations.

\begin{equation} \label{eqt:jcc_proof}
\sum_{\forall k \in \boldsymbol{K}(t_2,d_{i}(t_1))}{X^*_{i,k}}=c_{i}(t_2), \quad 2 \leq \forall i \leq N,
\end{equation}


\begin{equation} \label{eqt:pcc_proof}
\begin{split}
\sum_{\forall i \in \boldsymbol{J}(k,t_1)}{X^*_{i,k}} \leq Cap^*(W_k)=\left[ M-\sum_{ \forall i \notin \boldsymbol{J}(k,t_1)}{u_i}\right] \times l_k, \quad 2 \leq \forall k \leq K \\
\end{split}
\end{equation}


\begin{equation}  \label{eqt:ipc_proof}
X^*_{i,k} \leq l_k, \quad 2 \leq \forall i \leq N, \quad 2 \leq \forall k \leq K
\end{equation}

At $t_2$ when a new job $\tau_{1}(t_2)$ arrives, AJA($t_2$) is established using the fluid schedule. For AJA($t_2$) that covers the array of windows \{$W_2$,$W_3$,...$W_{K+1}$\}, we define three types of constraints as follows.


\begin{equation} \label{eqt:jcc1_proof}
\sum_{\forall k \in \boldsymbol{K}(t_2,d_{i}(t_2))}{X_{i,k}}=c_{i}(t_2), \quad 1 \leq \forall i \leq N
\end{equation}


\begin{equation} \label{eqt:pcc1_proof}
\begin{split}
\sum_{\forall i \in \boldsymbol{J}(k,t_2)}{X_{i,k}} \leq Cap'(W_k)=\left[ M-\sum_{ \forall i \notin \boldsymbol{J}(k,t_2)}{u_i}\right] \times l_k, \quad 2 \leq \forall k \leq K+1 \\
\end{split}
\end{equation}


\begin{equation}  \label{eqt:ipc1_proof}
X_{i,k} \leq l_k, \quad 1 \leq \forall i \leq N, \quad 2 \leq \forall k \leq K+1
\end{equation}

For AJA($t_2$), we select $\boldsymbol{X^{*}_{-1}} \cup \{X_{1,k}^{'} | 2\leq \forall k\leq K+1, \forall X_{1,k}^{'}=u_1l_k \}$ as a candidate solution. Now, we need to check if the candidate solution satisfies all constraints.

For JCCs, let Equation~\ref{eqt:jcc1_proof} be split as follows.



\begin{equation} \label{eqt:jcc2_proof}
=
\begin{cases}
\sum_{\forall k \in \boldsymbol{K}(t_2,d_{i}(t_2))}{X_{i,k}}=c_{i}(t_2), & \text{if } 2 \leq i \leq N \\
\sum_{\forall k \in \boldsymbol{K}(t_2,d_{1}(t_2))}{X_{1,k}}=C_{1}, & \text{if } i=1
\end{cases}
\end{equation}

Since $\boldsymbol{K}(t_2,d_{i}(t_2))=\boldsymbol{K}(t_2,d_{i}(t_1))$ when $2 \leq i \leq N$, Equation~\ref{eqt:jcc2_proof} becomes the following.

\begin{equation} \label{eqt:jcc3_proof}
=
\begin{cases}
\sum_{\forall k \in \boldsymbol{K}(t_2,d_{i}(t_1))}{X_{i,k}}=c_{i}(t_2), & \text{if } 2 \leq i \leq N \\
\sum_{\forall k \in \boldsymbol{K}(t_2,d_{1}(t_2))}{X_{1,k}}=C_{1}, & \text{if } i=1
\end{cases}
\end{equation}

The candidate solution satisfies the first case of Equation~\ref{eqt:jcc3_proof} since Equation~\ref{eqt:jcc_proof} holds. It also satisfies the second case of Equation~\ref{eqt:jcc3_proof} as follows.

\begin{equation} \label{eqt:jcc4_proof}
\sum_{\forall k \in \boldsymbol{K}(t_2,d_{1}(t_2))}{X'_{1,k}} = \sum_{\forall k \in \boldsymbol{K}(t_2,d_{1}(t_2))}{u_1l_k} = u_1P_1 = C_1
\end{equation}

Consequently, the candidate solution satisfies JCCs for AJA($t_2$).

For PCCs, the left-hand side of Inequality~\ref{eqt:pcc1_proof} is modified as follows. 

\begin{equation} \label{eqt:pcc2_proof}
\sum_{\forall i \in \boldsymbol{J}(k,t_2)}{X_{i,k}} = X_{1,k}+\sum_{\forall i \in \boldsymbol{J}(k,t_1)}{X_{i,k}} , \quad 2 \leq \forall k \leq K+1
\end{equation}

If the variables above are substituted with the candidate solution, it becomes the following due to Equation~\ref{eqt:pcc_proof}.

\begin{equation} \label{eqt:pcc3_proof}
\begin{split}
X'_{1,k}+\sum_{\forall i \in \boldsymbol{J}(k,t_1)}{X^*_{i,k}} \leq u_1l_k + Cap^*(W_k) = u_1l_k + \left[ M-\sum_{ \forall i \notin \boldsymbol{J}(k,t_1)}{u_i}\right] \times l_k \\
=\left[ M-\sum_{ \forall i \notin \boldsymbol{J}(k,t_1)}{u_i}+u_1\right] \times l_k \\
\end{split}
\end{equation}

\begin{equation}
\begin{split}
=
\begin{cases}
\left[ M-u_1+u_1\right] \times l_k & \text{,if }k=2 \\
\left[ M-(u_1+u_2)+u_1\right] \times l_k & \text{,if }k=3 \\
\text{...}&\\
\left[ M-(u_1+u_2+...+u_N)+u_1\right] \times l_k & \text{,if }k=K+1 \\
\end{cases}
\\
=\left[ M-\sum_{ \forall i \notin \boldsymbol{J}(k,t_2)}{u_i}\right]\cdot l_k = Cap'(W_k), \quad 2\leq \forall k\leq K+1&
\end{split}
\end{equation}

Therefore, the candidate solution satisfies PCCs.

In addition, since $\boldsymbol{X_{-1}^*}$ satisfies NIPs and $X'_{1,k}=u_1l_k \leq l_k$, $2 \leq \forall k \leq K+1$, the candidate solution satisfies NIPs.

In summary, the candidate solution satisfies all constraints of AJA($t_2$). It implies that if a feasible solution exists for AJA($t_1$), then at least one feasible solution for AJA($t_2$) can be found. The feasible solution for every AJA from time 0 to $H$ is ensured using the induction steps.
\end{proof}

\begin{corollary}
The proposed approach is an RT-optimal real-time scheduling algorithm.
\end{corollary}
\begin{proof}
According to Theorem~\ref{thm:optimality}, if the fluid schedule exists for a given task set, our technique provides the feasible solutions for the repeated AJAs. It is also known that if $U$ of the task set is less than or equal to $M$, RT-optimal fluid schedule exists. Therefore, if $U$ of a task set is less than or equal to $M$, our technique provides the feasible solutions for the repeated AJAs, which implies its RT-optimality.  
\end{proof}

\subsection{Flow network model}
\label{subsec:fnm}

The objective of the \textit{maximum flow} problem is to find the maximal flow from a single source to a single sink in the given flow network. To efficiently solve the LP problem for AJA, we transform it into the maximum flow problem by slightly changing JCCs as follows. 

\begin{equation} 
\label{eqt:jcc_fnm}
\sum_{\forall k \in \boldsymbol{K}(t_1,d_{i}(t_1))}{X_{i,k}} \leq c_{i}(t_1), \quad \forall i
\end{equation}

The maximum flow problem for AJA($t_1$) consists of the constraints, including Inequalities~\ref{eqt:jcc_fnm},~\ref{eqt:pcc_g} and~\ref{eqt:ipc_g}, and the objective function as follows.

\begin{definition} (Maximum flow problem for AJA($t_1$))
\label{def:max}

\begin{equation}  
\text{maximize} \quad u(\mathbf{X}) = \sum_{i=1}^{N}{ \sum_{\forall k \in \boldsymbol{K}(t_1,d_{i}(t_1))}{X_{i,k}} } .
\end{equation}

\begin{equation} 
\begin{split}
\text{s.t.}&\\
\sum_{\forall k \in \boldsymbol{K}(t_1,d_{i}(t_1))}{X_{i,k}} &\leq c_{i}(t_1), \quad 1 \leq \forall i \leq N,\\
\sum_{\forall i \in \boldsymbol{J}(k,t_1)}{X_{i,k}} &\leq Cap(W_k), \quad 1 \leq \forall k \leq K,\\
X_{i,k} &\leq l_k, \quad 1 \leq \forall i \leq N, 1 \leq \forall k \leq K.
\end{split}
\end{equation} 
\end{definition}

Note that when the maximum flow occurs in the given problem, the left-hand sides of all JCCs, i.e., Inequality~\ref{eqt:jcc_fnm}, should be equal to $c_i(t_1)$, which satisfies the original JCC equations. We call it the \textit{complete maximum flow}.  

Definition~\ref{def:max} turns our problem into a maximum flow problem. Since all constraints become the upper bounded inequalities, each upper bound acts as the capacity of each edge in the flow network. Using these edges, the flow network is built as follows.

To construct the flow network, we add two additional nodes, the source $n_s$ and the sink $n_e$. Between $n_s$ and $n_e$ in the network, two intermediate layers of nodes are placed, i.e., the first layer contains the nodes of all active jobs and the second layer contains the nodes of all windows. We name each node based on its corresponding job and window. Each edge is denoted by e($n_1$,$n_2$) where $n_1$ and $n_2$ are the source and destination nodes, respectively. $e(n_1,n_2) \leftarrow \{\sigma(n_1,n_2)\}$ denotes that the edge $e(n_1,n_2)$ has its flow capacity $\sigma(n_1,n_2)$.

From $n_s$ to node $\tau_{i}$, a directional edge is inserted and the flow capacity is determined as $c_{i}(t_1)$. From node $W_k$ to $n_e$, a directional edge is inserted with the capacity $Cap(W_k)$. Between node $\tau_{i}$ to node $W_k$, an edge whose flow is constrained by $l_{k}$ is inserted. 

Formally, when AJA($\cdot$) is given, it is transformed into a capacitated network $\mathbf{F}=(V,E)$, where a network $\mathbf{F}$ contains a set of nodes $V$ and a set of edges $E$. We call $\mathbf{F}$ the \textit{flow network for real-time scheduling} (FNRT).

\begin{equation}
\begin{split}
&V=\{n_s, n_e\} \cup \{\tau_{i} | \forall i \} \cup \{W_k | \forall k  \}\\
&E=\{e(n_s,\tau_{i})| \forall i \} \cup \{e(\tau_{i},W_k)| \forall i,k \} \cup \{e(W_k,n_e)| \forall k \} 
\end{split}
\end{equation}

The actual flow on an edge $(n_1,n_2)$ from node $n_1$ to node $n_2$ is denoted by $f(n_1,n_2)$. In FNRT, $f(\tau_{i}, W_k)$ corresponds to $X_{i,k}$ for the given AJA. In addition, the set of capacitated edges $\{e(n_s,\tau_{i})| \forall i \}$ represents JCCs. The sets of capacitated edges $\{e(\tau_{i},W_k)| \forall i,k \}$ and $\{e(W_k,n_e)| \forall k \}$ represent NIPs and PCCs, respectively. The complete maximum flow for FNRT satisfies all constraints even including the original JCCs in the linear programming problem. A flow network example for Figure~\ref{fig:step1} is shown in Figure~\ref{fig:example_implicit}.

\begin{figure}[h!]
\centerline{
	\includegraphics[scale=0.45]{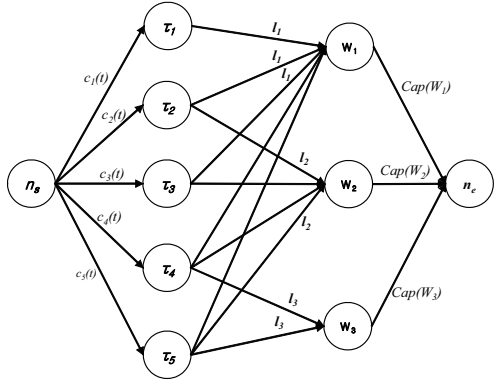}
	}
\caption{Flow network for the example}
\label{fig:example_implicit}
\end{figure}

\begin{algorithm}
\DontPrintSemicolon
\SetKwInOut{Input}{input}\SetKwInOut{Output}{output}
\Input{ the current time $t$ and the task set $\boldsymbol{\tau}$}
\KwData{ $\textbf{F}$ is the flow network to be constructed }
\KwData{ $\textbf{C}=\{Cap(W_k)|$ $k$ is for all $W_k$ in AJA $\}$ }
\KwData{ $e(n_1,n_2)$ is an edge from node $n_1$ to $n_2$ }
$\textbf{F} = null$ \;
$d^{max} =$ $max_{\forall i}\{d_i(t)\}$ \;
$\textbf{B}$ = \texttt{ObtainAllBoundaries}($t$,$d^{max}$) \;
$\textbf{C}$ = \texttt{ComputeWindowCapacities}($t$,$\textbf{B}$) \;
\For{$\tau_{i}(t) \in \boldsymbol{\tau}(t)$}{
	\For{$b_k \in \textbf{B}$}{
		\If {$b_k \in [t, d_{i}(t)]$}{
			Add $e(\tau_{i}, W_k)\leftarrow\{l_k\}$ to $\textbf{F}$ \;
		}
		\If	{$b_k == d_{i}(t)$}{
				Add $e(W_k, n_e)\leftarrow \{Cap(W_k)\}$ to $\textbf{F}$ \;
				break \;
		}
	}
	Add $e(n_s, \tau_{i})\leftarrow\{c_i(t)\}$ to $\textbf{F}$ \;
}
solve($\textbf{F}$) \;
\KwRet $\{X_{i,1}| 1 \leq i \leq N\}$ \;
\caption{\texttt{Schedule}($t$,$\boldsymbol{\tau}$) based on the flow network\label{alg:1}}
\end{algorithm}

A high-level description of the algorithm that is invoked at each scheduling event (boundary) is shown in Algorithm~\ref{alg:1}. Lines 1-12 shows how to construct FNRT. The maximum flow problem is solved by line 13. The computational complexity of constructing FNRT is proportional to the nested loops from lines 5-12 and though, the complexity of solving the maximum flow problem at line 13 is dominant.

The computational complexity of the algorithm depends primarily on both the number of nodes $|V|$ and the number of edges $|E|$. The number of nodes $|V|$ in the flow network is $|\{n_s, n_e\}| + |\{\tau_{i} | \forall i \}| + |\{W_k | \forall k  \}|$. The number of windows is less than or equal to the number of tasks. Thus, $|V|=2N+2$ in the worst case. The number of edges $|E|$ is $|\{e(n_s,\tau_{i})| \forall i \}| + |\{e(\tau_{i},W_k)| \forall i,k \}| + |\{e(W_k,n_e)| \forall k \}|$. In the worst case, $|\{e(n_s,\tau_{i})| \forall i \}| = |\{e(W_k,n_e)| \forall k \}|=N$ and $|\{e(\tau_{i},W_k)| \forall i,k \}|=1+2+...+N=N(N+1)/2$. Thus, $|E|=N^2/2+5N/2$. 

Maximum flow problem has been intensively studied in the graph theory research community and several polynomial algorithms have been found~\cite{gt2014}. In terms of complexity, the algorithms are categorized into two groups, i.e., weakly polynomial and strongly polynomial algorithms~\cite{gt2014}. The complexity of weakly polynomial algorithms is upper-bounded by a combination of $|V|$, $|E|$ and the largest capacity among all edges. On the other hand, the complexity of the strongly polynomial algorithms is upper-bounded by a combination of $|V|$ and $|E|$ alone. A strongly polynomial algorithm with $O(|V||E|)$ complexity~\cite{o2013} was introduced recently and further improvement continues. We use the $O(|V||E|)$ complexity algorithm for line 13 in Algorithm~\ref{alg:1}. Since $|V|$ and $|E|$ are proportional to $N$ and $N^2$ in the worst case respectively, the complexity of Algorithm~\ref{alg:1} is $O(N^3)$. Several other maximum flow algorithms with $O(|V|^3)$ complexity are possible alternatives~\cite{gt1986}. 

\subsection{fn-EDF on the continuous-time model}
\label{subsec:flow_control}

The objective of maximum flow algorithms is to send the maximal flow from the source to the sink. In general, multiple sets of flows that achieve the goal could exist, which implies that the formulated problem could have multiple solutions. Since maximum flow algorithms simply find one of the solutions depending on the preference of the selected algorithm, the amount of flow on each individual edge is not controlled, but only upper-bounded. In a real-time scheduling context, this implies that a task schedule determined by a maximum flow algorithm may be, for example, non-work-conserving or work-conserving from time to time. To a more carefully control the flow over FNRT, additional features of the flow network are required to use.

To control the flow over the network, we consider \textit{Minimum cost flow problem (MCFP)}. Each edge $e(n_1,n_2)$ on the flow network for MCFP contains one more parameter, called \textit{cost}, $w(n_1,n_2)$. When the actual flow $f$ is sent on an edge, the cost of the flow on the edge is calculated by $w \times f$. The objective of MCFP is to find the set of the actual flow $f$ on all edges that minimizes the total cost of the flow, when a total flow value is given. Since the total flow value is known to be the sum of the right-hand sides of JCCs for the given AJA, MCFP is easily applicable. The cost $w$ to each edge allows us to control the actual flows in the network. 

For example, assume that we try to generate an EDF-like schedule for the given AJA($t$). Among all edges $e(\tau_i,W_1)$ directed to node $W_1$, to prioritize the task execution in EDF order, we assign the lower cost to the edges $e(\tau_i,W_1)$ where $\tau_i$ has the earlier deadline. For example, the cost $\{1,...,N\}$ are assigned to all $e(\tau_i,W_1)$ in EDF order. 

\begin{equation}
\label{eqt:edge_cost_1}
e(\tau_i,W_1)\leftarrow\{l_1,i\}, \quad d_i(t) \leq d_{i+1}(t),
\end{equation}
where $e(n_1,n_2)\leftarrow\{\sigma(n_1,n_2),w\}$ denotes that the edge $e(n_1,n_2)$ has its capacity $\sigma(n_1,n_2)$ and cost $w$.  

In the following windows, to make the schedule work-conserving, the costs are increasingly assigned to edges, i.e., the cost $N+k-1$ is assigned to $e(\tau_i,W_k)$ where $2 \leq k \leq K$. 

\begin{equation}
\label{eqt:edge_cost_2}
e(\tau_i,W_k)\leftarrow\{l_k,N+k-1\}, \quad 2 \leq k \leq K.
\end{equation}

To implement this scheduling algorithm, line 8 in Algorithm~\ref{alg:1} should be updated with the previous two Equations,~\ref{eqt:edge_cost_1} and~\ref{eqt:edge_cost_2}, which yields Algorithm~\ref{alg:edf}. We call this \textit{flow network-based EDF} (fn-EDF).

\begin{figure}[]
\centering
\includegraphics[scale=0.45]{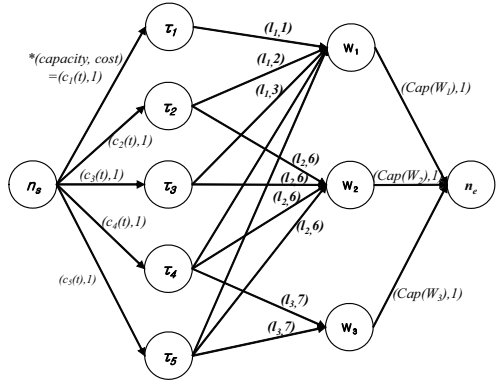}
\caption{Flow control}
\label{fig:example_implicit_cost}
\end{figure}

Figure~\ref{fig:example_implicit_cost} shows the flow network for fn-EDF that is extended from Figure~\ref{fig:example_implicit}.

To analytically describe the cost assignment, we introduce a notion, \textit{cost-slope}, $s_{i}$, at $W_k$ which is defined to be $w(\tau_i,W_{k+1})-w(\tau_i,W_{k})$. Assume that the costs $\{1,...,N\}$ are assigned to all $e(\tau_i,W_1)$ in EDF order and that the cost $N+1$ is assigned to all $e(\tau_i,W_2)$. Then, the earliest deadline task node has the highest cost-slope at $W_1$. We prove that the steeper cost-slope imposes the higher priority.

\begin{theorem}
When two job nodes $\tau_i$ and $\tau_j$ are connected to $W_k$ and $W_{k+1}$ with edges $\{$ $(\tau_i,W_k)$, $(\tau_i,W_{k+1})$, $(\tau_j,W_k)$, $(\tau_j,W_{k+1})$$\}$ in the flow network, if $s_i > s_j$ at $W_k$, then the flow network gives higher priority to the flow from $\tau_i$ to be sent to $W_k$ than that from $\tau_j$. 
\end{theorem}
\begin{proof}
Assume that one time unit (or flow) is available within $W_k$ and that $\tau_i$ and $\tau_j$ compete for it. In the case where $\tau_i$ occupies one time unit in $W_k$ and $\tau_j$ is delayed to take one time unit in $W^{k+1}$, we assume that the total flow cost is $u(\boldsymbol{F})$, where $\boldsymbol{F}$ is the set of actual flows on all edges. If the units of $\tau_i$ and $\tau_j$ are swapped between $W_k$ and $W_{k+1}$, the new cost $u(\boldsymbol{F'})$ becomes $u(\boldsymbol{F})+s_i-s_j$. Since $s_i > s_j$, $u(\boldsymbol{F'})>u(\boldsymbol{F})$. Therefore, MCFP algorithms prefer to assign the available single time unit in $W_k$ to $\tau_i$ rather than $\tau_j$ in order to minimize the total cost. All available time units are allocated in the same manner and thus, the flow from $\tau_i$ achieves a higher priority of being sent to $W_k$ than that from $\tau_j$.
\end{proof}

\begin{algorithm}
\DontPrintSemicolon
\SetKwInOut{Input}{input}\SetKwInOut{Output}{output}
\Input{ the current time $t$ and the task set $\boldsymbol{\tau}$}
\KwData{ $\boldsymbol{\tau^s}(t)$ is the active job set sorted by EDF }
\KwData{ $\tau_i$ in $\boldsymbol{\tau^s}(t)$ satisfies $d_i(t) \leq d_{i+1}(t)$ }

$\textbf{F} = null$ \;
$d^{max} =$ $max_{\forall i}\{d_i(t)\}$ \;
$\textbf{B}$ = \texttt{ObtainAllBoundaries}($t$,$d^{max}$) \;
$\textbf{C}$ = \texttt{ComputeWindowCapacities}($t$,$\mathbf{B}$) \;
$\boldsymbol{\tau^s}(t)$ = \texttt{sortByEDF}($\boldsymbol{\tau}(t)$) \;
\tcc{Assume $\tau_i$ in $\boldsymbol{\tau^s}(t)$ satisfies $d_i(t) \leq d_{i+1}(t)$} 
\For{$\tau_{i} \in \boldsymbol{\tau^s}(t)$}{
	\For{$b_k \in \mathbf{B}$}{
		\If {$b_k \in [t, d_{i}(t)]$}{
			\eIf {$k == 1$}{
				Add $e(\tau_{i}, W_1)\leftarrow\{l_1,i\}$ to $\mathbf{F}$ \;
			}{
				Add $e(\tau_{i}, W_k)\leftarrow\{l_k,N+k-1\}$ to $\mathbf{F}$ \;
			}
		}
		\If	{$b_k == d_{i}(t)$}{
				Add $e(W_k, n_e)\leftarrow \{Cap(W_k),1\}$ to $\mathbf{F}$ \;
				break \;
		}
	}
	Add $e(n_s, \tau_{i})\leftarrow\{c_i(t),1\}$ to $\mathbf{F}$ \;
}
solve($\mathbf{F}$) \;
\KwRet $\{X_{i,1}| 1 \leq i \leq N\}$ \;

\caption{\texttt{Schedule}($t$,$\boldsymbol{\tau}$) for flow control\label{alg:edf}}
\end{algorithm}

Note that the prioritization based on the cost slopes differs from the traditional prioritization in the real-time scheduling context. The assignment of traditional priorities to tasks may result in deadline misses, whereas prioritization using the cost slopes never attenuates the schedulability ensured by the flow network model. Thus, we can call this a \textit{weak priority} if distinction is needed. 

We also need to consider the maximum magnitude of costs, because the computational complexity of a certain class of MCFP algorithms depends on the cost. To generate the EDF-like schedule, the costs $\{$1, ..., $N+K$$\}$ are used. Since $K$ is proportional to $N$, the maximum magnitude of costs is proportional to $N$. We emphasize that different cost assignments are also possible for various scheduling purposes.

For the line 17 in Algorithm~\ref{alg:edf}, any minimum cost flow algorithm can be used. In terms of complexity, the algorithms are also categorized into two groups, i.e., weakly polynomial and strongly polynomial algorithms~\cite{k2015}. The complexity of the weakly polynomial algorithms is upper-bounded by a combination of $|V|$, $|E|$, the largest cost $w_{max}$ and (or) the largest capacity of all edges. We restrict $w_{max}$ for our scheduling purposes, e.g., $w_{max}\propto N$ for prioritizing task execution. Therefore, not only the strongly polynomial algorithms but also the polynomial algorithms bounded by a combination of $|V|$, $|E|$ and $w_{max}$ are of our interest. Goldberg \textit{et. al} introduced a \textit{cost-scaling algorithm with dynamic trees} having $O(|V||E|log|E|log(|E|w_{max}))$ complexity~\cite{gt1987}. Orlin introduced an \textit{enhanced capacity-scaling algorithm} having $O(|V|log|E|SP_+(|V|,|E|))$ complexity where $SP_+(|V|,|E|)$ denotes the time complexity of solving the single-source shortest path problem~\cite{o1993}. \textit{Dijkstra's algorithm with Fibonacci heaps} is known to provide an $O(|E|+|V|log|V|)$ bound for $SP_+(|V|,|E|)$~\cite{ft1987}. In our context, the complexity of the former algorithm is $O(N^3(logN)^2)$ and that of the latter is $O(N^3logN)$. 

\subsection{fn-EDF on the discrete-time model}
\label{subsec:discreteness}

One of our assumption in this study is that the task parameters, $\{C_i, P_i\}$, are multiples of the system time unit, i.e., integers. Despite the assumption, manipulating tasks for scheduling can generate non-integral numbers, which leads to two issues that we must consider.

First, several maximum flow and minimum cost flow algorithms assume that the parameters of the flow networks are integers. However, our flow network models permit non-integer network parameters. For example, in Algorithm~\ref{alg:1}, the edges $e(W_k,n_e)$ may have non-integral $Cap(W_k)$, which is calculated using $u_i=C_i/P_i$ of several tasks. Nevertheless, since $u_i$'s are rational numbers, this issue can be easily solved by linearly scaling all of the rational numbers in the flow network up to appropriate integers. It does not raise any concern about increasing complexity if a strongly polynomial algorithm is used. This simple technique is sufficient for scheduling tasks on the continuous-time model where the non-integral units of time are allowed to be allocated for task execution.

Second, on the discrete-time model, the allocation of the non-integral units of time for task execution is not allowed. To run tasks on the discrete-time model, all $X_{i,k}$'s of the flow network solution should be integers. 

RT-optimal scheduling on the discrete-time model has been studied by several researchers and BF algorithms were proposed recently~\cite{ghydvj2014}. To ensure RT-optimality, BF algorithms are designed to maintain their deviation from the fluid schedule to be less than one time unit at every boundary. Informally, when the fluid schedule yields a non-integral execution time $c'_i$ between two adjacent boundaries for a task, the algorithm divides the non-integral execution time into mandatory and optional executions where the mandatory execution time is determined to be $\lfloor{c'_i}\rfloor$ and the optional execution time is determined to be $c'_i-\lfloor{c'_i}\rfloor$. The optional execution times for all tasks are accumulated and additionally allocated to the tasks based on specific criteria.  

In Algorithms~\ref{alg:1} and~\ref{alg:edf}, $Cap(W_k)$ is calculated based on the fluid schedule of each task. For the discrete-time model, we instead use the BF schedule to determine $Cap(W_k)$ to be integral. Specifically, in Equation~\ref{eqt:cap_g}, the subtrahend $l_k \times \sum_{ \forall i \notin \boldsymbol{J}(k,t_1)}{C_i/P_i}$ is replaced with a proper integer value produced by BF schedule. 

Assume that BF invokes at the current time $t$. Then, BF returns a set $\mathbf{S_1}$ that contains the (integral) execution time $S_{i,1}$ of $\tau_i$ for the time interval $W_1$. The current time $t$ is assumed to be $b_0$ without loss of consistency with our previous notations.

\begin{equation}
\textbf{S}_{1} = \{S_{i,1} | 1 \leq \forall i \leq N\} = \texttt{runBF}(\mathbf{\tau}, b_0)
\end{equation}

When AJA($t$) is given, BF schedule within the time interval $[t,d^{max}]$ is obtained using the following equation.

\begin{equation}
\label{eqt:shadow_exe}
\mathbf{S}_{k} = \{S_{i,k} | 1 \leq \forall i \leq N\} = \texttt{runBF}(\mathbf{\tau}, b_{k-1}), 1 \leq  \forall k \leq K
\end{equation}

The fluid schedule on the continuous time model ensures that every task has a running rate $r_i=C_i/P_i$ all the time. On the other hand, BF on the discrete-time model has an individual rate $r_i$ that changes slightly for each $W_k$. It is because if $r_i \times l_k$ for $W_k$ is non-integral, it should be adjusted to be integral. Therefore, all boundaries including those made by future jobs within [$t$,$d^{max}(t)$] should be considered to calculate the integral execution time of each active job on the discrete-time model. Previously, the boundaries are defined as $\{d_i(t) | 1 \leq i \leq N\}$ for the continuous-time model. For the discrete-time model, the boundaries are determined as follows.

\begin{equation}
\label{eqt:shadow_b}
\mathbf{B} = \{ d_{i,j} | t \leq d_{i,j} \leq d^{max}(t), \forall i,j \}
\end{equation}

Based on these notions, the flow network-based scheduling algorithm for the discrete-time model runs as follows. When a scheduling event occurs at time $t$, the range of AJA($t$) is determined using $d^{max}(t)$. Within $[t,d^{max}(t)]$, all boundaries are determined using Equation~\ref{eqt:shadow_b}. Then, for every $W_k$ within $[t,d^{max}(t)]$, all $S_{i,k}$ are determined using Equation~\ref{eqt:shadow_exe}. Note that the determination of all $S_{i,k}$ requires consideration of the active jobs and the future jobs within $[t,d^{max}(t)]$ as well. These $S_{i,k}$ are used to calculate $Cap(W_k)$. In other words, BF runs in parallel with the actual flow network-based scheduling algorithm in order to calculate $Cap(W_k)$. It is implemented as the function \texttt{ComputeWindowCapacities}() as shown in Algorithm~\ref{alg:compute}.

\begin{algorithm}
\DontPrintSemicolon
\KwData { $t$ is the current time }
\KwData { $\mathbf{B}$ is the set of boundaries of all jobs in $[t, d^{max}]$ }
\For{$k=1$ \KwTo $K$}{
$\mathbf{S}_{k}$ = \texttt{runBF}($\mathbf{\tau}$, $b_{k-1}$) \;
$Cap(W_k)$ = $l_k \times M$ \;
\For{$\tau_i \in \mathbf{\tau}$}{
\If {$\tau_i$ is inactive in $W_k$}{
$Cap(W_k)$ = $Cap(W_k)$ - $S_{i,k}$ \;
}
}
}
\KwRet $\mathbf{C}=\{Cap(W_k)| 1 \leq k \leq K\}$ \;
\caption{\texttt{ComputeWindowCapacities}($t$,$\textbf{B}$) for discrete-time model\label{alg:compute}}
\end{algorithm}

Especially, we call the combination of Algorithm~\ref{alg:edf} and~\ref{alg:compute} \textit{fn-EDF on the discrete-time model}. To the best of our knowledge, fn-EDF is the first unfair-but-optimal scheduling algorithm for periodic implicit-deadline tasks on the discrete-time model. RT-optimality of fn-EDF is shown in the next section.

For fn-EDF, the number of boundaries of a task $\tau_i$ in the time interval $[t,d^{max}(t)]$, which additionally includes both $t$ and $d^{max}(t)$, is at most $\lceil(d^{max}(t)-t)/P_i\rceil+1$. The longest length of $[t,d^{max}(t)]$ is $P^{max}$ where $P^{max}=max_{\forall i}{\{P_i\}}$ and thus, the number of windows $K$ is proportional to $\sum_{\forall i}{\lceil P^{max}/P_i \rceil}$, where $\sum_{\forall i}{\lceil P^{max}/P_i \rceil}$ is denoted by $N'$. Since the complexity of \texttt{runBF()} is known to be $O(N \cdot P^{max})$, the complexity of \texttt{ComputeWindowCapacities()} is $O(N N' P^{max})$.

The increasing number of windows $K$ also affects the complexity of Algorithm~\ref{alg:edf}. The complexity of the dominant line 17 of Algorithm~\ref{alg:edf} is related to the number of nodes and edges in the flow network. The number of job nodes is $N$ and the number of window nodes is at most $N'$. Thus, $|V|$ is proportional to $N'$ and $|E|$ is proportional to $N \cdot N'$. Therefore, if the cost-scaling algorithm with dynamic trees is selected for MCFP, its complexity becomes $O(N^2 N' (logN')^2)$ in our context. 

Considering \texttt{ComputeWindowCapacities()}, the complexity of fn-EDF is $O(max\{ N N' P^{max}, N^2 N' (logN')^2 \})$. If the enhanced capacity-scaling algorithm is selected, its complexity becomes $O(N^2 N' logN' )$. Then, the complexity of fn-EDF is $O(max\{N N' P^{max}, N^2 N' logN' \})$.

Considering that $H$ is the least-common-multiple of all $P_i$, $P^{max}$ is usually assumed to be much smaller than $H$, which we also assume in this study. However, in the worst case when $P^{max}=H$, the complexity of the proposed algorithm becomes dependent on the number of all jobs in time $0$ to $H$. 


\subsection{RT-optimality of fn-EDF on the discrete-time model}

In order to prove RT-optimality of fn-EDF on the discrete-time model, we need some terminology for this section. $\mathbf{B}$ is defined to contain $a_{i,j}$ of all jobs in the time interval [$0$,$H$] and each element of $\mathbf{B}$ is denoted by $b_k$ where an earlier $b_k$ has a lower $k$. The allocated execution time of $\tau_{i,j}$ in $W_k$ is denoted by $X_{i,j,k}$. $\boldsymbol{J}(k, t)$ is extended as follows.

\begin{equation} 
\boldsymbol{J}(W_k) = \{ (i,j) | W_k \subset  [a_{i,j}, d_{i,j}] \}.
\end{equation}

As introduced in~\cite{msd2010}, we first construct a large FNRT for scheduling real-time tasks considering the long time interval [$0$,$H$] with the following constraints. 

\begin{definition} (Maximum flow problem for $\mathbf{\tau}$ within the time interval [$0$,$H$])
\begin{equation}  
\text{maximize} \quad u(\mathbf{X}) = \sum_{\forall W_k \subset [0,H]}{ \sum_{\forall (i,j) \in \boldsymbol{J}(W_k)}{X_{i,j,k}} } .
\end{equation}
s.t.
\begin{equation} \label{eqt:jcc_h}
\sum_{\forall k \in \boldsymbol{K}(a_{i,j},d_{i,j})}{X_{i,j,k}} \leq C_{i}, \quad \forall \tau_{i,j},
\end{equation}

\begin{equation} \label{eqt:pcc_h}
\sum_{\forall (i,j) \in \boldsymbol{J}(W_k)}{X_{i,j,k}} \leq l_k \times M, \quad \forall W_k \subset [0,H],
\end{equation}

\begin{equation}  \label{eqt:ipc_h}
X_{i,j,k} \leq l_k, , \quad \forall \tau_{i,j} \text{ and } \forall W_k \subset [0,H].
\end{equation}

\end{definition}

In FNRT as shown in $\mathbf{F}$, we assume that $\tau_{i,j}$ are sorted in the increasing arrival time order. For convenience of discussion, the sorted $\tau_{i,j}$ is labeled as $\tau_p$, $1 \leq p \leq P$. In addition, we assume that $W_k$ are sorted in the increasing starting time order. A set of edge flows, $\mathbf{f}$, is assumed to contain all edge flows. The amount of the maximum flow sent from $n_s$ to $n_e$ is denoted by $|\mathbf{f}|$. The complete maximum flow for $\mathbf{F}$ results in a feasible schedule where the flow $f(\tau_{i,j},W_k)$ corresponds to $X_{i,j,k}$. 



In the following lemma, we will define three flow networks $\mathbf{F}$, $\mathbf{F^+}$, and $\mathbf{F^-}$. The elements for each network is distinguished by using the superscript $+$ or $-$, e.g., $f(\cdot)$ denotes the flow for $\mathbf{F}$, $\sigma^+(\cdot)$ denotes an edge capacity for $\mathbf{F^+}$.

\begin{figure}[]
\centering
\includegraphics[scale=0.35]{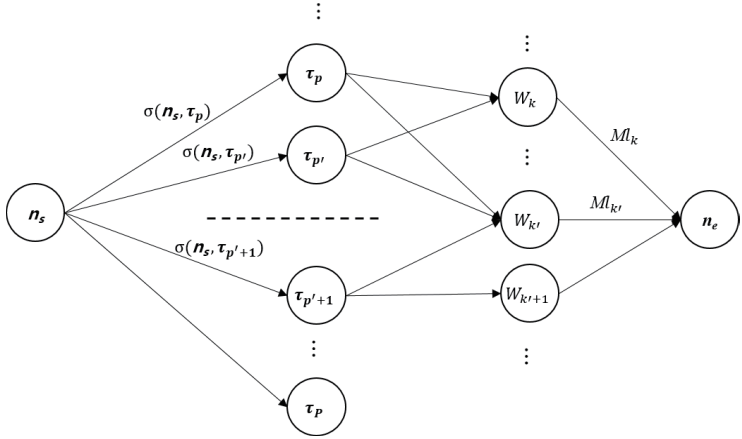}
\caption{A FNRT before decomposition}
\label{fig:split_1}
\end{figure}

\begin{figure}[]
\centering
\includegraphics[scale=0.35]{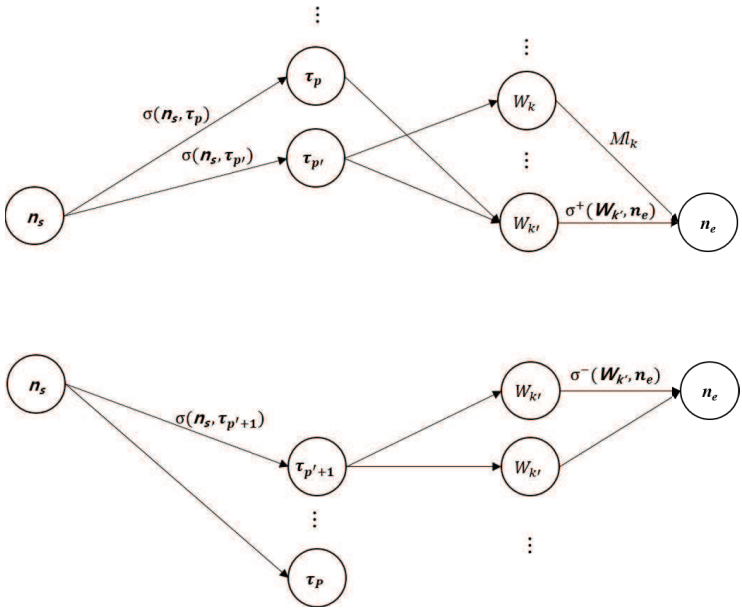}
\caption{A FNRT after decomposition}
\label{fig:split_2}
\end{figure}

\begin{lemma} \label{lemma:decomposition}
(Network decomposition) Assume that FNRT $\mathbf{F}$ is constructed for a task set $\mathbf{\tau}$ where the total utilization of $\mathbf{\tau}$ is less than equal to $M$. $\mathbf{F}$ is assumed to contain two disjoint task node sets, $\mathbf{\tau'}$ and $\mathbf{\tau''}$ where $\mathbf{\tau}=\mathbf{\tau'} \cup \mathbf{\tau''}$, $\mathbf{\tau'}=\{\tau_p|1 \leq p \leq p'\}$, and $\mathbf{\tau''}=\{\tau_p|p'+1 \leq p \leq P\}$ as shown in Figure~\ref{fig:split_1}. Then, the maximum flow $\mathbf{f}$ for $\mathbf{F}$ consists of two flows, $\mathbf{f'}$ sent from $n_s$ through $\mathbf{\tau'}$ to $n_e$ and $\mathbf{f''}$ sent from $n_s$ through $\mathbf{\tau''}$ to $n_e$. 

When $W_{k'} = Child(\mathbf{\tau'}) \cap Child(\mathbf{\tau''})$, suppose that $F$ is decomposed into $\mathbf{F^+}$ and $\mathbf{F^-}$ where $\mathbf{F^+}$ contains $\mathbf{\tau'}$ and their child nodes and $\mathbf{F^-}$ contains $\mathbf{\tau''}$ and their child nodes as shown in Figure~\ref{fig:split_2}. If the edge capacities $\sigma(W_{k'},n_e)$ for $\mathbf{F^+}$ and $\mathbf{F^-}$ are determined as follows, 
\begin{equation} \label{eqt:lemma_decomposition_cap}
\begin{split}
&\sigma^{-}(W_{k'},n_e) = \sum_{\forall \tau_p \in \mathbf{\tau''} \cap Parent(W_{k'})}{f(\tau_p,W_{k'})}, \text{ and} \\
&\sigma^{+}(W_{k'},n_e) = M l_k - \sum_{\forall \tau_p \in \mathbf{\tau''} \cap Parent(W_{k'})}{f(\tau_p,W_{k'})},\\
\end{split}
\end{equation}
then, the maximum flow for $\mathbf{F^+}$ is $|\mathbf{f'}|$ and the maximum flow for $\mathbf{F^-}$ is $|\mathbf{f''}|$.
\end{lemma}

\begin{proof}
In $\mathbf{F}$, 
\begin{equation}
\begin{split} \label{eqt:split}
f(W_{k'},n_e)&=\sum_{\forall \tau_p \in Parent(W_{k'})}{f(\tau_p,W_{k'})}\\
&=\sum_{\forall \tau_p \in Parent(W_{k'}) \cap \mathbf{\tau'}}{f(\tau_p,W_{k'})}+\sum_{\forall \tau_p \in Parent(W_{k'}) \cap \mathbf{\tau''}}{f(\tau_p,W_{k'})} \\
&=f(\tau',W_{k'}) + f(\tau'',W_{k'}) \leq Ml_k.\\
\end{split}
\end{equation}
In $\mathbf{F^-}$ after decomposition, since $\sigma^{-}(W_{k'},n_e)$ is sufficient to send $f(\tau'',W_{k'})$ along $e^-(W_{k'},n_e)$ and all the other edge capacities are preserved, the maximum flow $|\mathbf{f^-}|$ is not lower than $|\mathbf{f''}|$.

The maximum flow $|\mathbf{f^-}|$ is not greater than $|\mathbf{f''}|$, neither. First, the flow through the node $W_{k'}$ is upper-bounded by the constraint of the edge capacity $\sigma^{-}(W_{k'},n_e)$. Second, if the other flows detouring the node $W_{k'}$ increase in $\mathbf{F^-}$ compared to those in $\mathbf{F}$, it implies that the flow $\mathbf{f''}$ in $\mathbf{F}$ is not the maximum flow, which contradicts the assumption that $\mathbf{f}$ is the maximum flow for $\mathbf{F}$. Therefore, $|\mathbf{f''}|$ is the maximum flow for $\mathbf{F^-}$.

In $\mathbf{F^+}$ after decomposition, the following inequality holds for the edge capacity based on Equation~\ref{eqt:split}.

\begin{equation}
\begin{split}
\sigma^{+}(W_{k'},n_e)&= M l_k - \sum_{\forall \tau_p \in Parent(W_{k'}) \cap \mathbf{\tau''}}{f(\tau_p,W_{k'})} \\
&\geq \sum_{\forall \tau_p \in Parent(W_{k'}) \cap \mathbf{\tau^-}}{f(\tau_p,W_{k'})} 
\end{split}
\end{equation}

Thus, $\sigma^{+}(W_{k'},n_e)$ is sufficient to send the flow $f(\tau'',W_{k′})$ along $e^+(W_{k'},n_e)$ and all the other edge capacities are preserved, the maximum flow $|\mathbf{f^+}|$ is not lower than $|\mathbf{f'}|$. 

The maximum flow $|\mathbf{f^+}|$ is not greater than $|\mathbf{f'}|$, neither. If $|\mathbf{f^+}|$ is greater than $|\mathbf{f'}|$, it implies that the flow $\mathbf{f'}$ in $\mathbf{F}$ is not the maximum flow, which contradicts the assumption that $\mathbf{f}$ is the maximum flow for $\mathbf{F}$. Therefore, $|\mathbf{f'}|$ is the maximum flow for $\mathbf{F^+}$.
\end{proof}

Lemma~\ref{lemma:decomposition} only considers the case that the set $Child(\mathbf{\tau'}) \cap Child(\mathbf{\tau''})$ contains a single window node $W_{k'}$. If it contains more nodes, then Lemma~\ref{lemma:decomposition} can be easily extended by applying Equation~\ref{eqt:lemma_decomposition_cap} to each $W_{k'}$.

\begin{figure}[]
\centering
\includegraphics[scale=0.35]{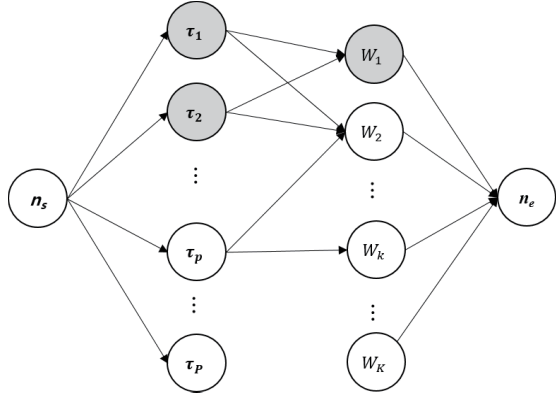}
\caption{A FNRT for Lemma~\ref{lemma:node_removal}}
\label{fig:node_removal}
\end{figure}

\begin{lemma} \label{lemma:node_removal}
(Node removal) Assume that FNRT $\mathbf{F}$ is given with its maximum flow $\mathbf{f}$ as shown in Figure~\ref{fig:node_removal}. If an earliest window node $W_1$ with its parent nodes $Parent(W_1)$ is removed and the edge capacities of $\mathbf{F}$ are adjusted as follows, 
\begin{equation}
\sigma'(n_s,\tau_p) = \sigma(n_s,\tau_p) - f(\tau_p,W_1), \text{ where } \forall \tau_p \in Parent(W_1) 
\end{equation}
then, $\mathbf{F}$ turns into $\mathbf{F'}$ with its maximum flow $|\mathbf{f'}|=|\mathbf{f}|-\sum_{\forall \tau_p \in Parent(W_1)}{|f(\tau_p,W_1)|}$.
\end{lemma}
\begin{proof}
If $|\mathbf{f'}|$ is not the maximum flow for $\mathbf{F'}$, the increased flow can be sent to $n_e$ in $\mathbf{F'}$. It implies that the increased flow can be sent through the other nodes $W_k$ than $W_1$ in $\mathbf{F}$, which contradicts our assumption that $\mathbf{f}$ is the maximum flow for $\mathbf{F}$.
\end{proof}

\begin{lemma} \label{lemma:composition}
(Network composition) Assume that two FNRTs, $\mathbf{F^+}$ and $\mathbf{F^-}$, are given with their complete maximum flows $\mathbf{f^+}$ and $\mathbf{f^-}$ respectively. In addition, assume that both $\mathbf{F^+}$ and $\mathbf{F^-}$ have their own window node $W_{k'}$. If 
two flow networks are combined to be $\mathbf{F}$ by merging the node $W_{k'}$ and adjusting the edge capacity as follows,
\begin{equation}
\sigma(W_{k'},n_e) = \sigma^{+}(W_{k'},n_e) + \sigma^{-}(W_{k'},n_e),
\end{equation}
then $\mathbf{F}$ has its complete maximum flow $|\mathbf{f^+}|+|\mathbf{f^-}|$. 
\end{lemma}
\begin{proof}
In $\mathbf{F}$, all edge capacities are preserved excluding $\sigma(W_{k'},n_e)$ on $e(W_{k'},n_e)$. $\sigma(W_{k'},n_e)$ is large enough to send $f^+(W_{k'},n_e)$ and $f^-(W_{k'},n_e)$ and thus, $|\mathbf{f}|$ is not lower than $|\mathbf{f^+}|+|\mathbf{f^-}|$. 
$|\mathbf{f}|$ is not greater than $|\mathbf{f^+}|+|\mathbf{f^-}|$, neither. Since both $\mathbf{f^+}$ and $\mathbf{f^-}$ are the complete maximum flow, $|\mathbf{f}|$ is upper-bounded by the edge capacities $\sum_{\forall \tau_p}{\sigma(n_s,\tau_p)}$.
\end{proof} 

The following two propositions are straightforward by definition of FNRT.

\begin{proposition} \label{lemma:sched_fn}
If a real-time scheduler $\mathscr{A}$ feasibly schedules a task set $\mathbf{\tau}$ during the time interval [$0$,$H$], its schedule $X_{i,j,k}$ corresponds to the flow $f(\tau_{i,j},W_k)$ of the complete maximum flow for FNRT $\mathbf{F}$ constructed for [$0$,$H$].
\end{proposition}

\begin{proposition} \label{lemma:fn_sched}
When FNRT $\mathbf{F}$ is constructed for a task set $\mathbf{\tau}$ during [$0$,$H$], the complete maximum flow for $\mathbf{F}$ contains the flow $f(\tau_{i,j},W_k)$ corresponding to the feasible schedule $X_{i,j,k}$.
\end{proposition}

\begin{theorem} \label{thm:BF_fnEDF_feasible}
If the feasible BF schedule exists for a given task set $\mathbf{\tau}$, then fn-EDF on the discrete-time model provides a feasible solution. 
\end{theorem}
\begin{proof}
It is proved by induction on increasing boundaries from $0$ to $H$. When a task set $\mathbf{\tau}$ is given, BF generates a feasible schedule by determining the execution time $S_{i,j,k}$ allocated for each $\tau_{i,j}$ in $W_k$, where $W_k$ is the time window placed in [$0$,$H$].  \\

At time 0, according to Proposition~\ref{lemma:sched_fn}, if BF provides a feasible schedule, FNRT $\mathbf{F}$ for [$0$,$H$] can be constructed with its complete maximum flow $\mathbf{f}$. $\mathbf{f}$ includes $f(\tau_{i,j},W_k)$ that corresponds to $S_{i,j,k}$ of BF. 

According to Lemma~\ref{lemma:decomposition}, $\mathbf{F}$ can be decomposed into two flow networks $\mathbf{F^a}$ and $\mathbf{F^b}$ where $\mathbf{F^a}$ contains the active job node set $\tau^a$ within the time interval $W_1$, i.e., $\{\tau_{i,j}|\forall \tau_{i,j} \in Parent(W_1)\}$, and $\mathbf{F^b}$ contains the rest job nodes. When $\mathbf{F}$ is decomposed, the shared edges $\{e(W_{k'},n_e) | W_{k'} \in \cup_{\forall \tau_{i,j} \in \mathbf{\tau^a}}{Child(\tau_{i,j})}\}$ are split to $\mathbf{F^a}$ and $\mathbf{F^b}$ and their capacities are determined by Lemma~\ref{lemma:decomposition}. 

\begin{equation}
\begin{split}
\sigma^{b}(W_{k'},n_e) &= \sum_{\forall \tau_p \in \mathbf{\tau^b} \cap \in Parent(W_{k'})}{f(\tau_p,W_{k'})} \\
&= \sum_{\forall \tau_p \in \mathbf{\tau^b} \cap \in Parent(W_{k'})}{S_{i,j,k'}},\\
\sigma^{a}(W_{k'},n_e) &= M l_k - \sum_{\forall \tau_p \in \mathbf{\tau^b} \cap \in Parent(W_{k'})}{f(\tau_p,W_{k'})} \\
&= M l_k - \sum_{\forall \tau_p \in \mathbf{\tau^b} \cap \in Parent(W_{k'})}{S_{i,j,k'}}, \\
\end{split}
\end{equation}

Then, both $\mathbf{F^a}$ and $\mathbf{F^b}$ have their complete maximum flow $\mathbf{f^a}$ and $\mathbf{f^b}$ by Lemma~\ref{lemma:decomposition}.

Assume that fn-EDF finds a maximum flow $\mathbf{f'^a}$ for $\mathbf{F^a}$, where $|\mathbf{f'^a}|=|\mathbf{f^a}|$. Note that all $\sigma^a(W_{k'},n_e)$ for $\mathbf{F^a}$ are still preserved and $\mathbf{f'^a}$ is upper-bounded by the capacities. Then, based on Lemma~\ref{lemma:node_removal}, the first window node and its parent nodes can be removed from $\mathbf{F^a}$, where we call the FNRT having the first window node removed $\mathbf{F'^a}$. $\mathbf{F'^a}$ has its complete maximum flow $|\mathbf{f'^a}|-|f'^a(W_1,n_e)|$. 

Then, $\mathbf{F'^a}$ with its complete maximum flow $|\mathbf{f'^a}|-|f'^a(W_1,n_e)|$ and $\mathbf{F^b}$ with its complete maximum flow $|\mathbf{f^b}|$ can be combined based on Lemma~\ref{lemma:composition}. The edge capacities $\sigma(W_{k'},n_e)$ then comes to $M l_k$. The combined FNRT has its complete maximum flow $|\mathbf{f'^a}|-|f'^a(W_1,n_e)|+|\mathbf{f^b}|$. 

Since the final FNRT has its complete maximum flow, according to Proposition~\ref{lemma:fn_sched}, at least a feasible schedule exists in the time interval [$W_2$,$H$]. In addition, note that $\mathbf{f^b}$ corresponding BF schedule is still preserved. Therefore, the same procedure above can be iteratively conducted up to $H$. 
\end{proof}

\begin{corollary} 
fn-EDF on the discrete-time model is an RT-optimal real-time scheduling algorithm.
\end{corollary}
\begin{proof}
According to Theorem~\ref{thm:BF_fnEDF_feasible}, if BF schedule exists for a given task set, fn-EDF provides a feasible solution. It is also known that if $U$ of the task set is less than or equal to $M$, BF schedule exists. Therefore, if $U$ of a task set is less than or equal to $M$, fn-EDF provides the feasible solutions, which implies its RT-optimality.  
\end{proof}

\section{Experiments}
\label{sec:experiments}

We experimentally evaluated the performance of fn-EDF compared with BF on the discrete-time model, as BF is the latest technique in the problem domain. \textit{SimSo} was used as the simulation environment~\cite{chd2014}. We performed four sets of experiments for 2, 4, 6 and 8 processors where the task sets were randomly generated. The experimental procedure was the following.

\begin{itemize}
\setlength\itemsep{0.01em}
\item Each task set contained the following number of tasks, $\{2 M, 2.5 M, 3 M, 3.5 M, 4 M\}$.
\item The unit time was set to 1.
\item For a task set, $P_i$'s were randomly generated using the uniform distribution within [5, 20].
\item If $H$ was greater than its upper-bound, $600,000$, $P_i$'s were discarded and this procedure restarted. The upper-bound helped avoid long time experiments. 
\item The total utilization $U$ for tasks was set to $M$.
\item $u_i$'s were randomly generated based on the fixed $U$ using the algorithm in \cite{s2006}. The algorithm generates random numbers within their predefined ranges when their fixed sum is given. 
\item Based on each $P_i$ and $u_i$, $C_i$ was calculated by $max\{1.0,\left\lfloor u_i \times P_i \right\rfloor\}$. Then, all $C_i$ were determined to be integral and the final $U$ became slightly less than or equal to $M$. 
\item We ran both BF and fn-EDF for the task set.
\item We repeated this procedure until more than 100 task sets were tested for each result. 
\end{itemize}

\begin{figure}
	\label{fig:experiments}
	\centering
	\subfloat[Preemptions on 2 processors]{
		\includegraphics[scale=0.35]{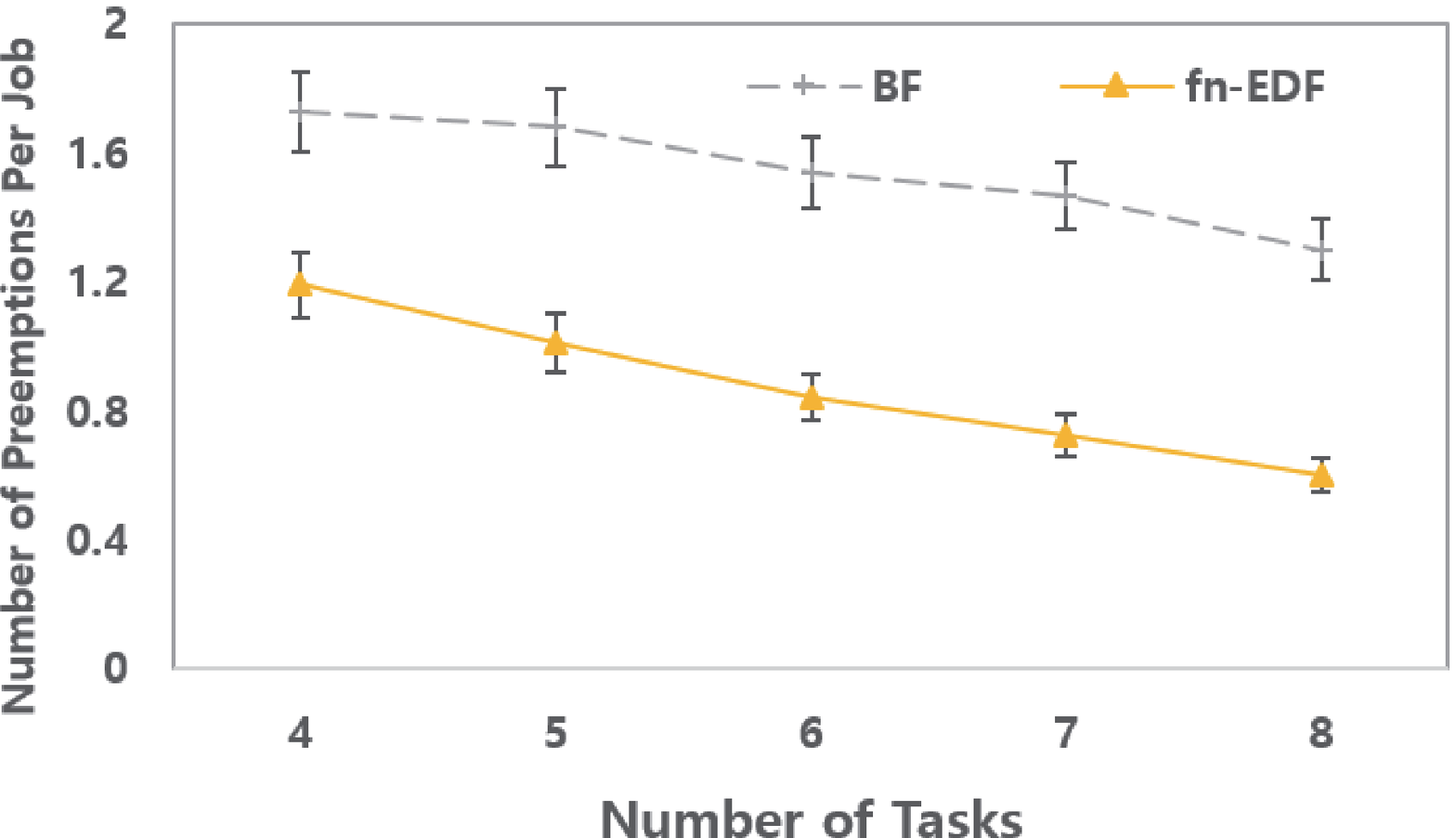}
	}
	\subfloat[Migrations on 2 processors]{		
		\includegraphics[scale=0.35]{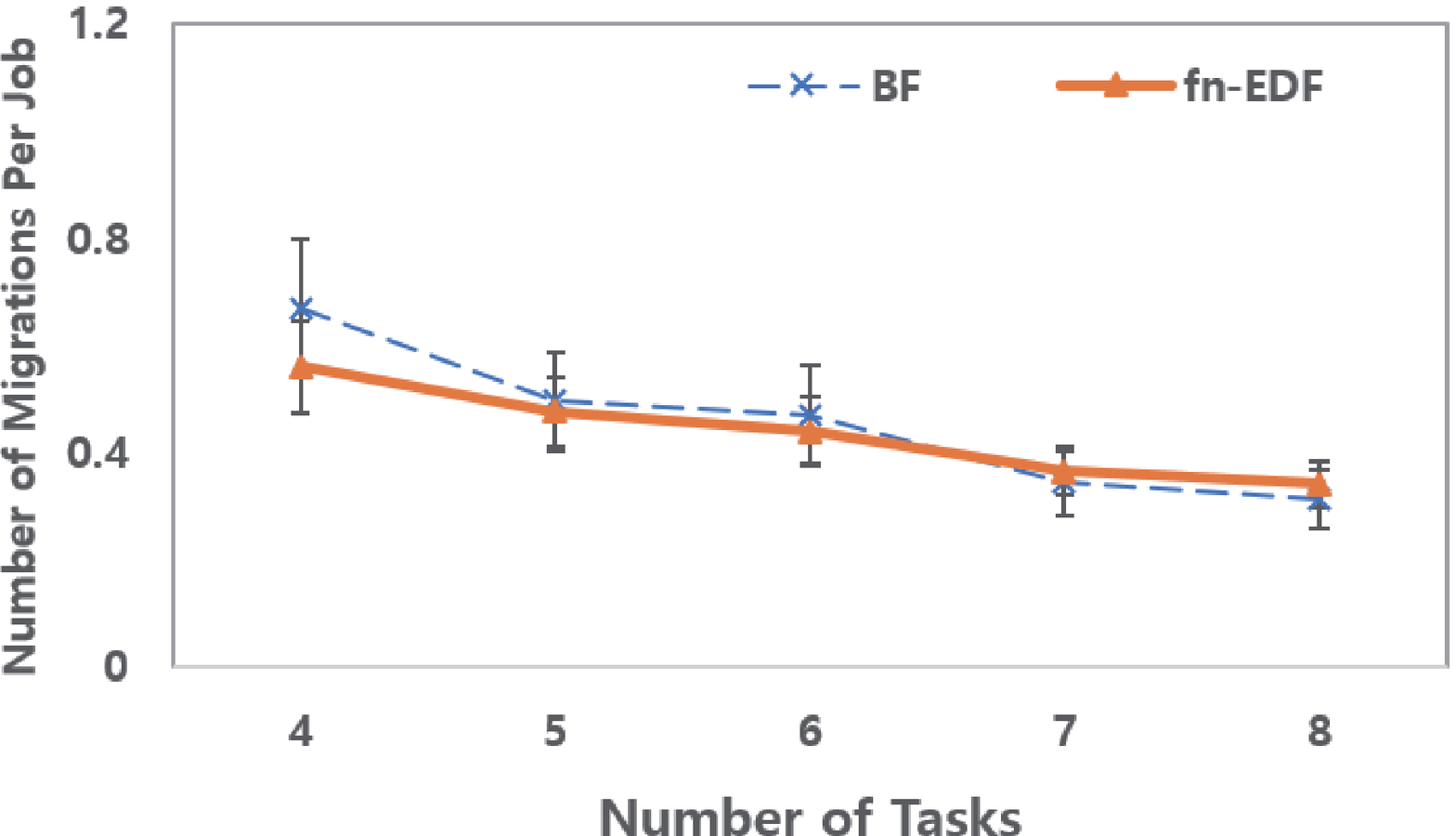}
	}\\
	\subfloat[Preemptions on 4 processors]{
		\includegraphics[scale=0.35]{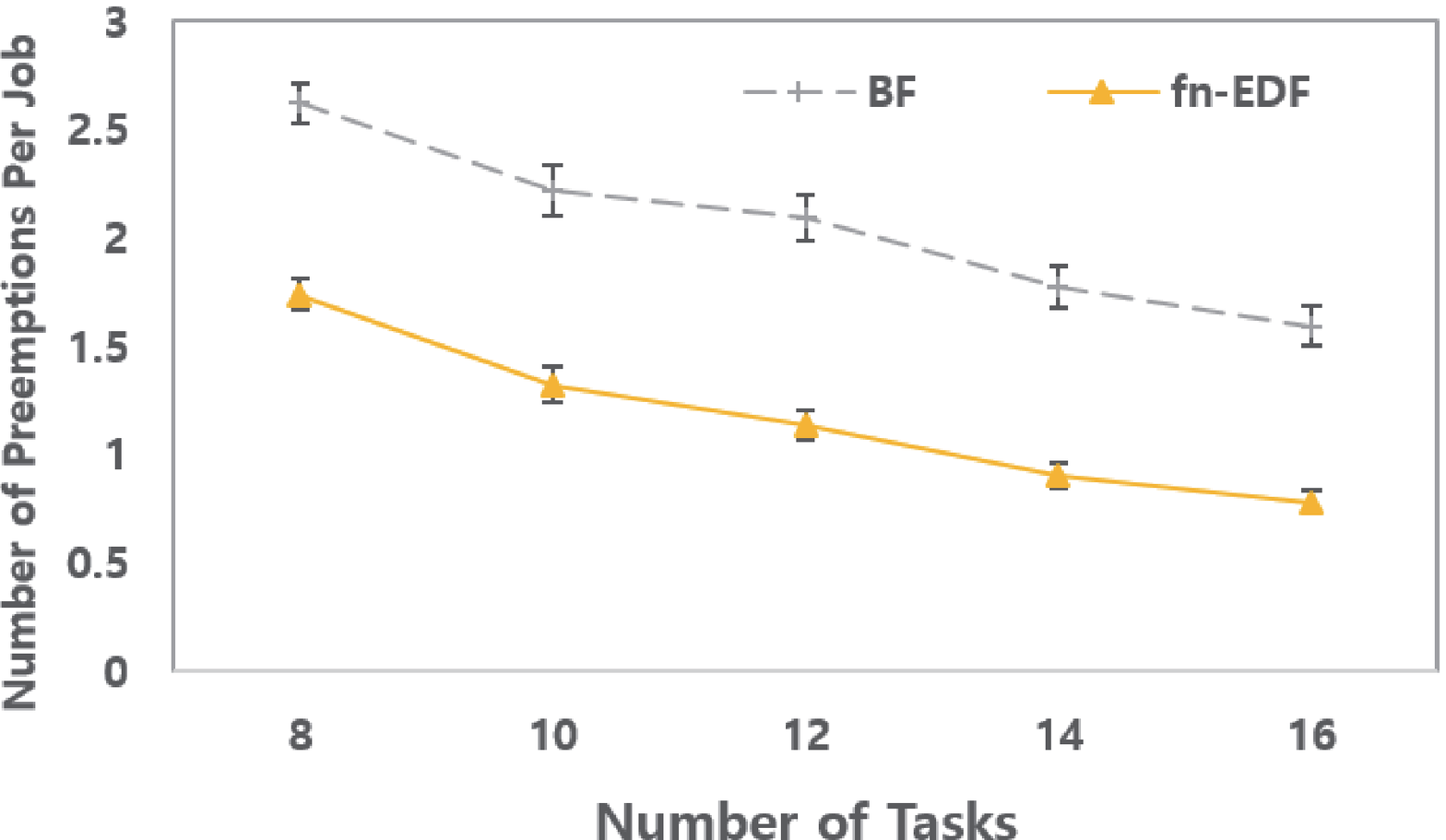}
	}
	\subfloat[Migrations on 4 processors]{		
		\includegraphics[scale=0.35]{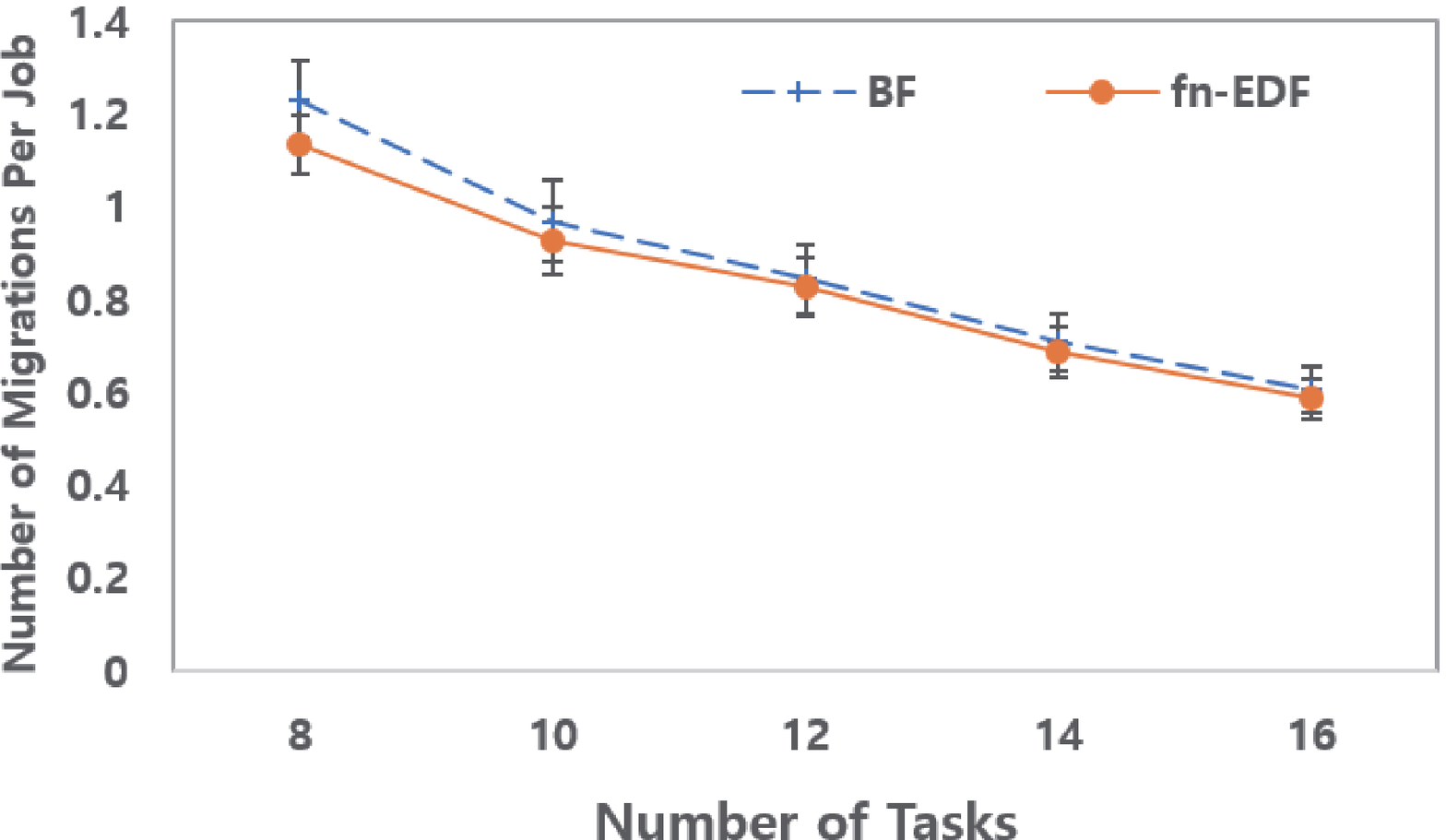}
	}\\
	\subfloat[Preemptions on 6 processors]{
		\includegraphics[scale=0.35]{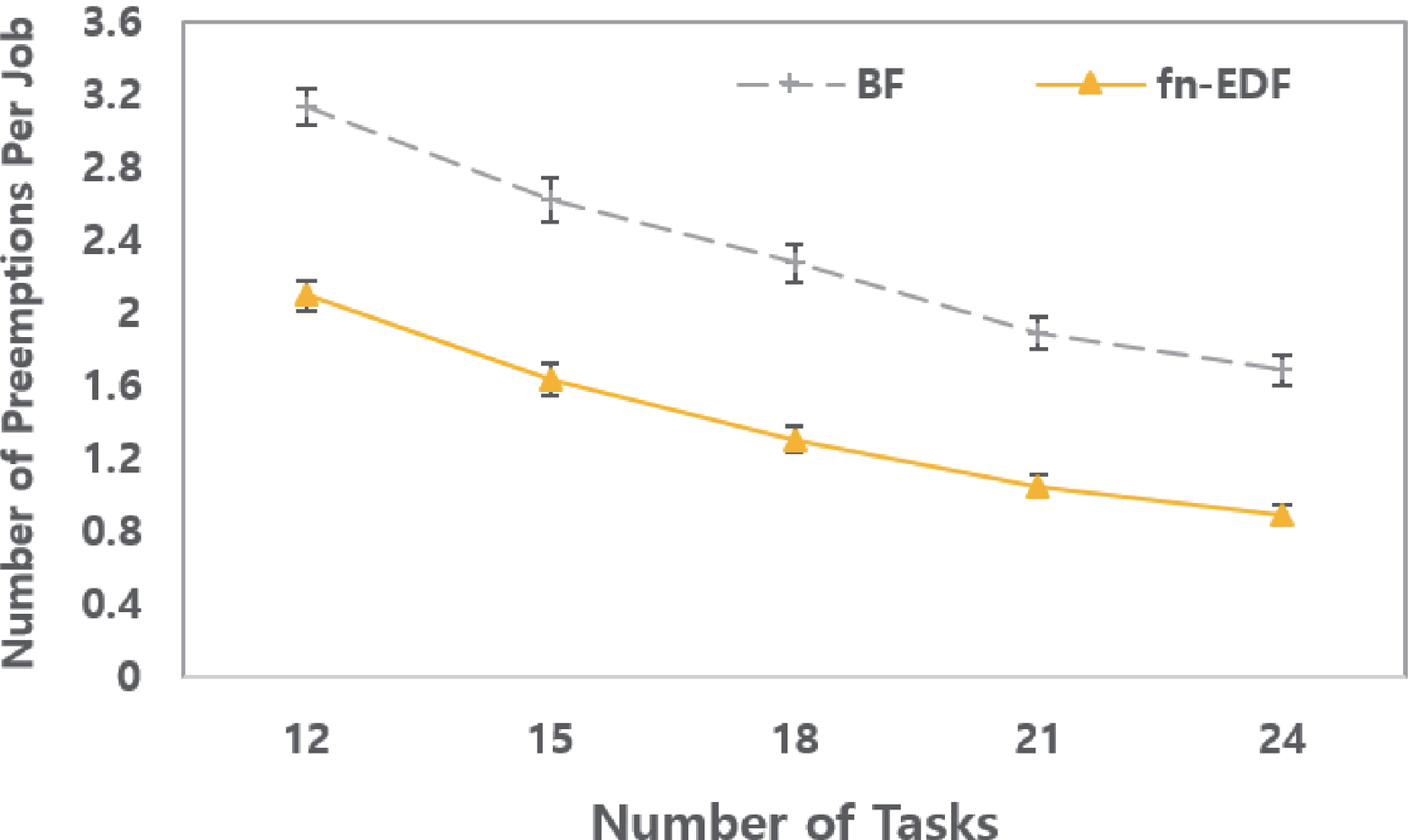}
	}
	\subfloat[Migrations on 6 processors]{		
		\includegraphics[scale=0.35]{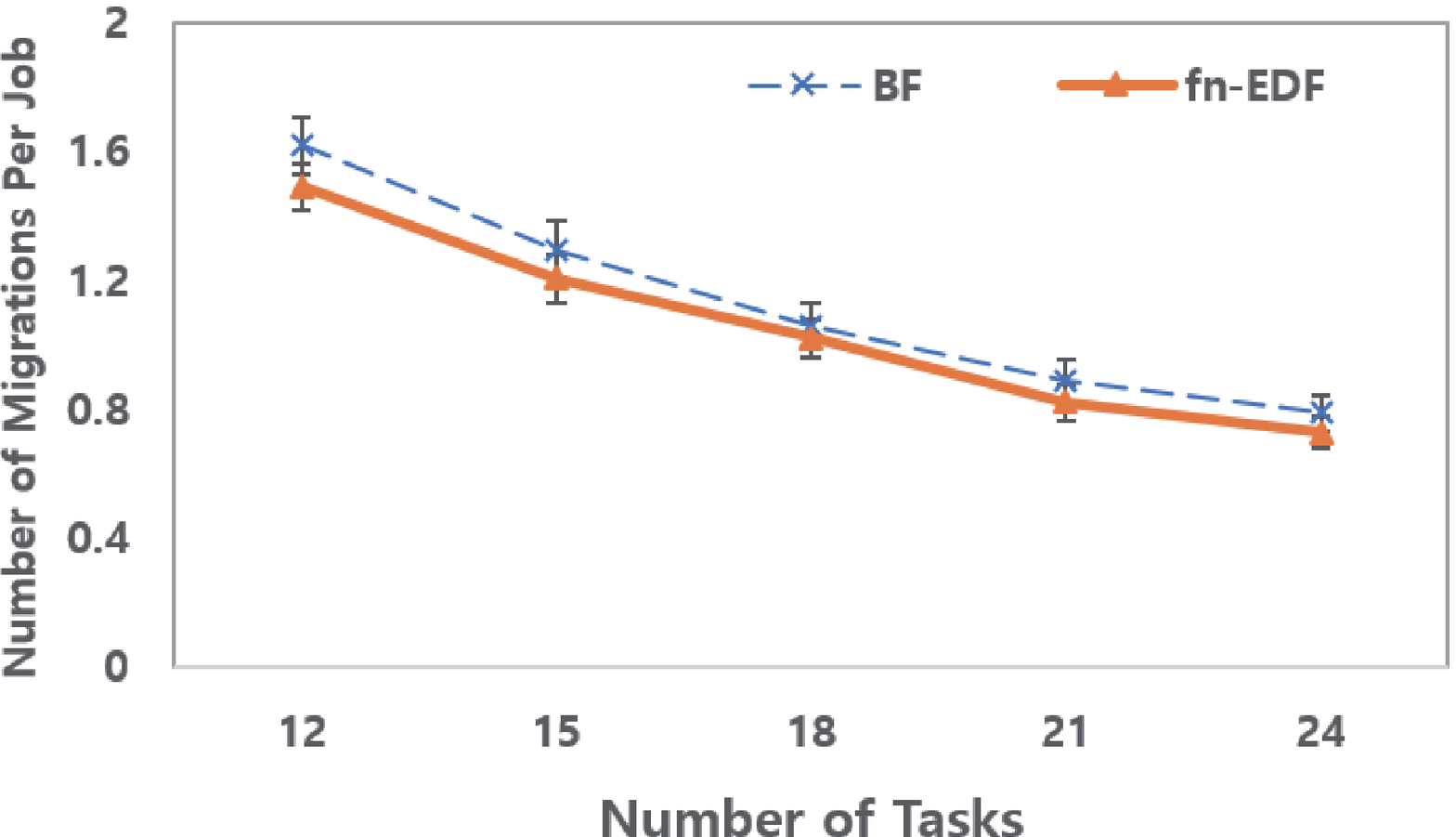}
	}\\		
	\subfloat[Preemptions on 8 processors]{
		\includegraphics[scale=0.35]{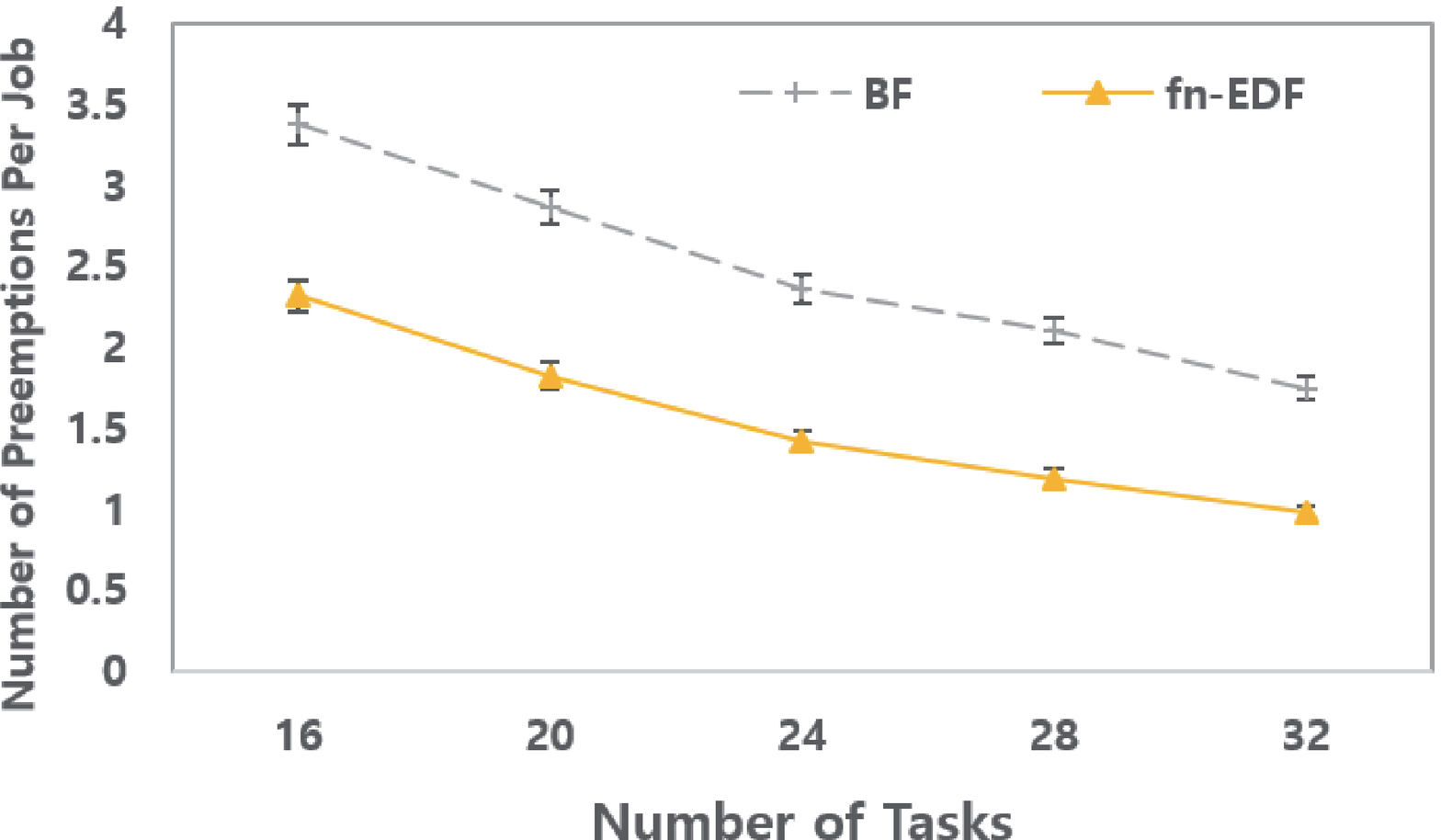}
	}
	\subfloat[Migrations on 8 processors]{		
		\includegraphics[scale=0.35]{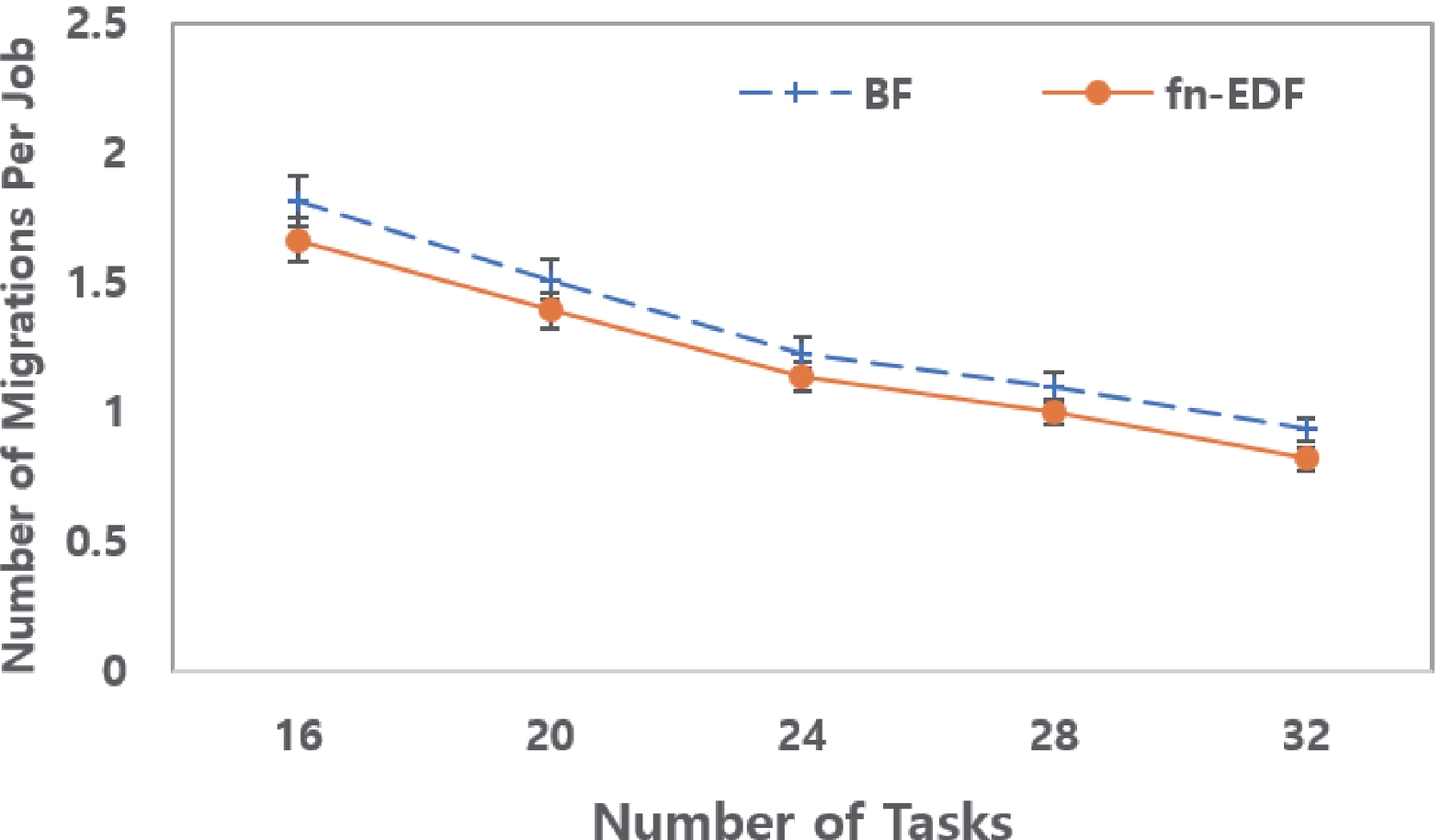}
	}		
	\caption{Number of preemptions and migrations per job. The error bar around each data point represents 95\% confidence interval of that data point.}
\end{figure}

Figure 12 shows \textit{the average numbers of preemptions and migrations per job} with BF and fn-EDF. The figures show that the number of preemptions with fn-EDF is always less than that with BF, while the number of migrations with fn-EDF is comparable to or slightly less than that with BF. Moreover, the \textit{preemption reduction ratio} of fn-EDF over BF, i.e., the number of preemptions with fn-EDF divided by that with BF, becomes more significant, when more tasks run on the system. For example, it reaches around $0.5$ when $4M$ tasks run on $M$ processors. It is caused by the difference between BF and fn-EDF. BF generally forces every task to consume its execution time for each window, whereas fn-EDF allows a task to skip its execution in some windows, which significantly reduces the number of preemptions. 

In the current implementation of fn-EDF, the lower indexed processor is simply allocated to the active job with the earlier deadline. In order to reduce the number of preemptions and migrations further for fn-EDF, the processor should be allocated to jobs more carefully by an improved heuristic, that is one of our future works.

\section{Discussion}
\label{sec:discussion}

In Section~\ref{subsec:flow_control} and~\ref{subsec:discreteness}, we have discussed the computational complexities of both versions of fn-EDF. Although the proposed algorithms have the polynomial complexity, the computationally lighter algorithms should be pursued. We believe that further reduction of their computational complexity is promising for two reasons. First, the organization of the flow network models in this study is rather fixed, e.g., four layers of nodes including two intermediate layers of job and window nodes. Therefore, there might exist an improved algorithm specialized for our flow network models. For example, if the source and sink nodes are aligned with the window and the job nodes respectively while retaining their connections, the flow network becomes a bipartite network having a single source and a single sink node. A bipartite network is defined as a network where the nodes are partitioned into two subsets and each edge has one endpoint in the one subset and the other endpoint in the other subset~\cite{amo1993}. Thus, the maximum flow algorithms specialized for the bipartite networks are applicable to our network models. Second, research effort in the graph theory community to achieve better maximum flow algorithms continues. For example, one of the conclusion given by~\cite{gt2014} states that an $O(|V||E|/|V|^{\epsilon})$ strongly polynomial algorithm may exist. As we have established the flow network models for the online multiprocessor real-time scheduling problems, we have the opportunity to directly utilize the advance in the modern graph theory.

We also explore the possibility of extending the proposed approach for the real-time tasks having different arriving patterns, e.g., sporadic tasks. On the continuous-time model, fn-EDF can be easily extended for sporadic tasks by letting each task own its portion $u_i$ out of the system's total processing capacity $M$, even when the task is inactive as proposed in~\cite{lfspb2010}. On the discrete-time model, extending fn-EDF for sporadic tasks is not so simple. It is because fn-EDF on the discrete-time model is based on BF that has each task's portion (or the rate $r_i$) out of the system's total processing capacity slightly changing over windows. Thus, it is difficult to determine how much capacity portion should be given to a task when it is inactive. Extending fn-EDF for sporadic tasks is one of our future works.



\section{Related work}
\label{sec:related}

There have been several attempts to formulate the multiprocessor real-time scheduling problem as an LP problem. Since these attempts considered all jobs of each task from time $0$ to the hyper-period for problem formulation, they have been used as the preliminary offline technique to compensate for the shortcomings of the subsequent online algorithms~\cite{msd2010,ljp2013}.  Lawler \textit{et al.} formulated the multiprocessor constrained-deadline periodic task scheduling as an LP problem for the first time~\cite{ec1981}. Transforming this linear programming problem into a network flow problem was considered in~\cite{msd2010}.

In addition, several RT-optimal solutions have been proposed for the multiprocessor implicit-deadline periodic task scheduling problem. Most of them rely on the notion of fairness which was first proposed by Baruah \textit{et al.}~\cite{bcpv1993,bcpv1996}. The perfect fairness can be obtained by the RT-optimal but non-realistic fluid schedule and the fairness-based RT-optimal scheduling algorithms were designed to follow the fluid schedule as much as possible. The first algorithm of this kind was \textit{PFfair} scheduler that invokes every time quantum to restrict the deviation of resource allocation for each task from its fluid schedule within one time quantum. Several variants were proposed later such as \textit{pseudo-deadline} (PD)~\cite{bgp1995}, PD$^2$~\cite{as2001}, etc. However, Pfair and its variants are known to lead to frequent preemption and migration of tasks, which motivated several alternatives to arise.

Zhu \textit{et al.} introduced \textit{Boundary fair} (BF) algorithm that invokes every boundary to restrict the deviation of resource allocation for each task from its fluid schedule within one unit time~\cite{zmm2003}. Since the scheduler-invocation frequency of BF is lower than that of PFair, the scheduling overheads are reduced. Recently, \textit{BF$^2$} algorithm was proposed as an extension of BF to handle sporadic tasks~\cite{ghydvj2014}. While BF and BF$^2$ were based on the discrete-time model, there was another branch of algorithms which supported a similar boundary fairness based on the continuous-time model~\cite{crj2006,flspb2011}. Readers who are interested in the fairness-based RT-optimal scheduling algorithms can refer to~\cite{ra2011,ghydvj2014}. Although these boundary fairness-based scheduling algorithms often outperform the proportionate fairness-based scheduling in terms of reducing the number of preemptions and migrations, the need for further minimizing the scheduling overhead remains.

Nelissen \textit{et al.} observed that respecting the fairness causes the high overheads~\cite{gvjd2011}. Thus, they proposed \textit{U-EDF}, an unfair but still RT-optimal scheduling algorithm, and they showed that U-EDF performs better than the existing fairness-based algorithms in reducing the scheduling overhead. U-EDF is known to be based on the continuous-time model. 



\section{Conclusion}
\label{sec:conclusion}

In this study, we formulated the multiprocessor real-time task scheduling problems by identifying three types of constraints and suggested the flow network models to solve the formulated problems efficiently. We discussed the potential use of this framework for addressing several interesting scheduling problems without losing the scheduling optimality. Based on the framework, the unfair-but-optimal scheduling algorithm, fn-EDF, was proposed for both continuous and discrete-time models. Our experiments showed that fn-EDF outperforms an existing BF algorithm in terms of the number of preemptions while maintaining a comparable number of task migrations. 

Formulating the multiprocessor real-time task scheduling problem by identifying its constraints is beneficial for expanding its possible applications. Simply by manipulating the constraints, this formulation can be flexibly extended to solve other related problems. One of the potential applications include the example of managing the processor idle time, shown in Section~\ref{sec:motiv_examples}. In addition to this horizontal expansion of its application, several existing techniques for the flow networks can be adapted for specific real-time scheduling purposes, e.g., \textit{load balancing} techniques that aim at balancing the flow on the selected edges of a given network~\cite{ph2005}. Our future research will consider these issues.

%
%

\begin{acknowledgements}
This research was supported by Basic Science Research Program through the National Research Foundation of Korea(NRF) funded by the Minist (NRF-2015R1D1A1A01057018).
\end{acknowledgements}

\bibliographystyle{authordate1}


\begin{thebibliography}{1}


\bibitem[\protect\citename{Ahuja {\em et~al.}, }1993]{amo1993}
Ahuja, R.~K., Magnanti, T.~L., \& Orlin, J.~B. 1993.
\newblock {\em Network Flows: Theory, Algorithms, and Applications}.
\newblock Prentice Hall.

\bibitem[\protect\citename{Anderson \& Srinivasan, }2001]{as2001}
Anderson, J.H., \& Srinivasan, A. 2001.
\newblock Mixed PFair/ERfair scheduling of asynchronous periodic tasks.
\newblock {\em In:} {\em Proceedings of the 13th Euromicro Conference on
  Real-Time Systems}.

\bibitem[\protect\citename{Baruah {\em et~al.}, }1993]{bcpv1993}
Baruah, S.~K., Cohen, N.~K., Plaxton, C.~G., \& Varvel, D.~A. 1993.
\newblock Proportionate Progress: A Notion of Fairness in Resource Allocation.
\newblock {\em In:} {\em Proceedings of the 25th Annual ACM Symposium on Theory
  of Computing}.

\bibitem[\protect\citename{Baruah {\em et~al.}, }2015]{bbb2015}
Baruah, S.~K., Bertogna, M., \& Buttazzo, G. 2015.
\newblock {\em Multiprocessor Scheduling for Real-Time Systems}.
\newblock Springer International Publishing.

\bibitem[\protect\citename{Baruah {\em et~al.}, }June 1996]{bcpv1996}
Baruah, S.~K., Cohen, N.~K., Plaxton, C.~G., \& Varvel, D.~A. June 1996.
\newblock Proportionate progress: A notion of fairness in resource allocation.
\newblock {\em Algorithmica}, {\bf 15}(6), 600--625.

\bibitem[\protect\citename{Buttazzo, }2011]{b2011}
Buttazzo, G. 2011.
\newblock {\em Hard Real-Time Computing Systems: Predictable Scheduling
  Algorithms and Applications}.
\newblock Springer.

\bibitem[\protect\citename{Cho {\em et~al.}, }2006]{crj2006}
Cho, H., Ravindran, B., \& Jensen, E.~D. 2006.
\newblock An Optimal Real-Time Scheduling Algorithm for Multiprocessors.
\newblock {\em In:} {\em Proceedings of the 27th IEEE Real-Time Systems
  Symposium}.

\bibitem[\protect\citename{Chéramy {\em et~al.}, }2014]{chd2014}
Chéramy, M., Hladik, P., \& Déplanche, A. 2014.
\newblock SIMSO: A Simulation Tool to Evaluate Real-Time Multiprocessor
  Scheduling Algorithms.
\newblock {\em In:} {\em Euromicro Conference on Real-Time Systems,
  Work-in-Progress}.

\bibitem[\protect\citename{Davis \& Burns, }Oct. 2011]{ra2011}
Davis, R.~I., \& Burns, A. Oct. 2011.
\newblock A Survey of Hard Real-time Scheduling for Multiprocessor Systems.
\newblock {\em ACM Computing Surveys}, {\bf 43}(4), 35:1--35:44.

\bibitem[\protect\citename{Funaoka {\em et~al.}, }2008]{fky2008}
Funaoka, K., Kato, S., \& Yamasaki, N. 2008.
\newblock Work-Conserving Optimal Real-Time Scheduling on Multiprocessors.
\newblock {\em In:} {\em Euromicro Conference on Real-Time Systems}.

\bibitem[\protect\citename{Funk {\em et~al.}, }Jul. 2011]{flspb2011}
Funk, S., Levin, G., Sadowski, C., Pye, I., \& Brandt, S. Jul. 2011.
\newblock DP-FAIR: a unifying theory for optimal hard real-time multiprocessor
  scheduling.
\newblock {\em Real-Time Systems}, {\bf 47}(5), 389--429.

\bibitem[\protect\citename{Goldberg \& Tarjan, }1986]{gt1986}
Goldberg, A.~V., \& Tarjan, R.~E. 1986.
\newblock A new approach to the maximum flow problem.
\newblock {\em In:} {\em Proceedings of the eighteenth annual ACM symposium on
  Theory of computing}.

\bibitem[\protect\citename{Goldberg \& Tarjan, }1987]{gt1987}
Goldberg, A.~V., \& Tarjan, R.~E. 1987.
\newblock Solving minimum-cost flow problems by successive approximation.
\newblock {\em In:} {\em Proceedings of 19th ACM Symposium on Theory of
  Computing}.

\bibitem[\protect\citename{Goldberg \& Tarjan, }Aug. 2014]{gt2014}
Goldberg, A.~V., \& Tarjan, R.~E. Aug. 2014.
\newblock Efficient Maximum Flow Algorithms.
\newblock {\em Communications of the ACM}, {\bf 57}(8), 82--89.

\bibitem[\protect\citename{Goldberg \& Tarjan, }Feb. 2015]{k2015}
Goldberg, A.~V., \& Tarjan, R.~E. Feb. 2015.
\newblock Minimum-cost flow algorithms: an experimental evaluation.
\newblock {\em Journal Optimization Methods \& Software}, {\bf 30}(1), 94--127.

\bibitem[\protect\citename{Horn, }Mar. 1974]{h1974}
Horn, W.A. Mar. 1974.
\newblock Some Simple Scheduling Algorithms.
\newblock {\em Naval Research Logistics}, {\bf 21}(1), 177–185.

\bibitem[\protect\citename{Lawler \& Martel, }Feb.1981]{ec1981}
Lawler, E.~L., \& Martel, C.~U. Feb.1981.
\newblock Scheduling Periodically occurring tasks on multiple processors.
\newblock {\em Information Processing Letters}, {\bf 12}(1), 9–12.

\bibitem[\protect\citename{Legout {\em et~al.}, }2013]{ljp2013}
Legout, V., Jan, M., \& Pautet, L. 2013.
\newblock A Scheduling Algorithm to Reduce the Static Energy Consumption of
  Multiprocessor Real-Time Systems.
\newblock {\em In:} {\em Proceedings of the 21st International conference on
  Real-Time Networks and Systems}.

\bibitem[\protect\citename{Levin {\em et~al.}, }2010]{lfspb2010}
Levin, G., Funk, S., Sadowski, C., Pye, I., \& Brandt, S. 2010.
\newblock DP-FAIR: A Simple Model for Understanding Optimal Multiprocessor
  Scheduling.
\newblock {\em In:} {\em 22nd Euromicro Conference on Real-Time Systems}.

\bibitem[\protect\citename{Liu, }1969]{l1969}
Liu, C.~L. 1969.
\newblock Scheduling algorithms for multiprocessors in a hard real-time
  environment.
\newblock {\em JPL Space Programs Summary}, {\bf 37}(60), 28--31.

\bibitem[\protect\citename{M.~L.~Fredman, }Jul. 1987]{ft1987}
M.~L.~Fredman, R. E.~Tarjan. Jul. 1987.
\newblock Fibonacci heaps and their uses in improved network optimization
  algorithms.
\newblock {\em Journal of the ACM}, {\bf 34}(3), 596--615.

\bibitem[\protect\citename{McNaughton, }1959]{m1959}
McNaughton, R. 1959.
\newblock Scheduling with deadlines and loss functions.
\newblock {\em Management Science}, {\bf 6}(1), 1--12.

\bibitem[\protect\citename{Megel {\em et~al.}, }2010]{msd2010}
Megel, T., Sirdey, R., \& David, V. 2010.
\newblock Minimizing task preemptions and migrations in multiprocessor optimal
  real-time schedules.
\newblock {\em In:} {\em Proceedings of the 31st IEEE Real-Time Systems
  Symposium}.

\bibitem[\protect\citename{Nelissen {\em et~al.}, }2011]{gvjd2011}
Nelissen, G., Berten, V., Goossens, J., \& Milojevic, D. 2011.
\newblock Reducing preemptions and migrations in real-time multiprocessor
  scheduling algorithms by releasing the fairness.
\newblock {\em In:} {\em IEEE International Conference on Embedded and
  Real-time Computing Systems and Applications}.

\bibitem[\protect\citename{Nelissen {\em et~al.}, }2012]{nbngm2012}
Nelissen, G., Berten, V., Nelis, V., Goossens, J., \& D.Milojevic. 2012.
\newblock U-EDF: An Unfair But Optimal Multiprocessor Scheduling Algorithm for
  Sporadic Tasks.
\newblock {\em In:} {\em 24th Euromicro Conference on Real-Time Systems}.

\bibitem[\protect\citename{Nelissen {\em et~al.}, }Jul. 2014]{ghydvj2014}
Nelissen, G., Su, H., Guo, Y., Zhu, D., Nelis, V., \& Goossens, J. Jul. 2014.
\newblock An optimal boundary fair scheduling.
\newblock {\em Real-Time Systems}, {\bf 50}(4), 456--508.

\bibitem[\protect\citename{Orlin, }2013]{o2013}
Orlin, J.~B. 2013.
\newblock Max Flows in O(nm) Time, or Better.
\newblock {\em In:} {\em Proceedings of the 45th Annual ACM Symposium on Theory
  of Computing}.

\bibitem[\protect\citename{Orlin, }1993]{o1993}
Orlin, J.B. 1993.
\newblock A faster strongly polynomial minimum cost flow algorithm.
\newblock {\em Operations Research}, {\bf 41}(2), 338--350.

\bibitem[\protect\citename{Pinar \& Hendrickson, }Oct. 2005]{ph2005}
Pinar, A., \& Hendrickson, B. Oct. 2005.
\newblock Improving Load Balance with Flexibly Assignable Tasks.
\newblock {\em IEEE Transactions on Parallel and Distributed Systems}, {\bf
  16}(10).

\bibitem[\protect\citename{Regnier {\em et~al.}, }2011]{rlmlb2011}
Regnier, P., Lima, G., Massa, E., Levin, G., \& Brandt, S. 2011.
\newblock RUN: Optimal Multiprocessor Real-Time Scheduling via Reduction to
  Uniprocessor.
\newblock {\em In:} {\em Proceedings of the 32th IEEE International Real-Time
  Systems Symposium}.

\bibitem[\protect\citename{S.K.Baruah, }1995]{bgp1995}
S.K.Baruah, J.~Gehrke, C.G.~Plaxton. 1995.
\newblock Fast scheduling of periodic tasks on multiple resources.
\newblock {\em In:} {\em Proceedings of The International Parallel Processing
  Symposium}.

\bibitem[\protect\citename{Srinivasan {\em et~al.}, }2003]{shab2003}
Srinivasan, A., Holman, P., Anderson, J.~H., \& Baruah, S.~K. 2003.
\newblock The case for fair multiprocessor scheduling.
\newblock {\em In:} {\em Proceedings of the 17th International Symposium on
  Parallel and Distributed Processing}.

\bibitem[\protect\citename{Stafford, }2006]{s2006}
Stafford, R. 2006.
\newblock Random vectors with fixed sum. [Online].
\newblock {\em In:} {\em Available:
  http://www.mathworks.com/matlabcentral/fileexchange/9700}.

\bibitem[\protect\citename{Zhu {\em et~al.}, }2003]{zmm2003}
Zhu, D., Mosse, D., \& Melhem, R. 2003.
\newblock Multiple-resource periodic scheduling problem: how much fairness is
  necessary?
\newblock {\em In:} {\em Proceedings of the 24th IEEE International Real-Time
  Systems Symposium}.

\bibitem[\protect\citename{Zhu {\em et~al.}, }2011]{zqmm2011}
Zhu, D., Qi, X., Mosse, D., \& Melhemb, R. 2011.
\newblock An optimal boundary fair scheduling algorithm for multiprocessor
  real-time systems.
\newblock {\em Journal of Parallel and Distributed Computing}, {\bf 71}(10),
  1411--1425.

\end{thebibliography}

%
%

\end{document}